\documentclass[nohyperref]{article}

\usepackage{microtype}
\usepackage{graphicx}
\usepackage{booktabs} 

\usepackage{hyperref}

\usepackage[accepted]{icml2022}

\usepackage{amsmath}
\usepackage{amssymb}
\usepackage{mathtools}
\usepackage{amsthm}
\usepackage{subcaption}

\usepackage[capitalize,noabbrev]{cleveref}

\usepackage{mathtools}
\usepackage{graphicx}
\usepackage{tikz}
\usetikzlibrary{bayesnet}
\usetikzlibrary{arrows, arrows.meta, shapes, calc, intersections, positioning, decorations.pathreplacing, fit, backgrounds, patterns}

\newcommand{\mathbbm}[1]{\text{\usefont{U}{bbm}{m}{n}#1}}

\theoremstyle{plain}
\newtheorem{theorem}{Theorem}[section]
\newtheorem{proposition}[theorem]{Proposition}

\newtheorem{corollary}[theorem]{Corollary}
\newtheorem{assumption}[theorem]{Assumption}

\theoremstyle{definition}
\newtheorem{definition}[theorem]{Definition}
\newtheorem{remark}[theorem]{Remark}
\newtheorem{example}{Example}

\usepackage[textsize=tiny]{todonotes}

\newcommand{\cA}{\mathcal{A}}
\newcommand{\cB}{\mathcal{B}}

\newcommand{\cC}{\mathcal{C}}
\newcommand{\cO}{\mathcal{O}}

\newcommand{\cF}{\mathcal{F}}
\newcommand{\cG}{\mathcal{G}}
\newcommand{\cD}{\mathcal{D}}
\newcommand{\cU}{\mathcal{U}}
\newcommand{\cX}{\mathcal{X}}
\newcommand{\cY}{\mathcal{Y}}

\newcommand{\cN}{\mathcal{N}}

\newcommand{\cV}{\mathcal{V}}
\newcommand{\cL}{\mathcal{L}}

\newcommand{\bE}{\mathbb{E}}

\newcommand{\indep}{\perp \!\!\! \perp}
\newcommand{\sign}{\text{sign}}
\newcommand{\An}{\text{An}}
\newcommand{\Pa}{\text{Pa}}
\newcommand{\De}{\text{De}}
\newcommand{\Var}{\text{Var}}
\newcommand{\Cov}{\text{Cov}}
\newcommand{\MAP}{\scriptsize \text{MAP}}

\icmltitlerunning{Causal Bias Quantification for Continuous Treatments}

\begin{document}

\twocolumn[
\icmltitle{Causal Bias Quantification for Continuous Treatments}



\icmlsetsymbol{equal}{*}

\begin{icmlauthorlist}
\icmlauthor{Gianluca Detommaso}{equal,comp}
\icmlauthor{Michael Brückner}{equal,comp}
\icmlauthor{Philip Schulz}{comp}
\icmlauthor{Victor Chernozhukov}{comp,yyy}
\end{icmlauthorlist}

\icmlaffiliation{comp}{Amazon, Berlin, Germany}
\icmlaffiliation{yyy}{Department of Economics, Massachusetts Institute of Technology, USA}

\icmlcorrespondingauthor{Gianluca Detommaso}{detomma@amazon.de}
\icmlcorrespondingauthor{Michael Brückner}{brueckm@amazon.de}

\icmlkeywords{Causal Inference}

\vskip 0.3in
]



\printAffiliationsAndNotice{\icmlEqualContribution} 

\begin{abstract}
We extend the definition of the marginal causal effect to the continuous treatment setting and develop a novel characterization of causal bias in the framework of structural causal models. We prove that our derived bias expression is zero if, and only if, the causal effect is identifiable via covariate adjustment. We show that under some restrictions on the structural equations, the causal bias can be estimated efficiently and allows for causal regularization of predictive probabilistic models. We demonstrate the effectiveness of our method for causal bias quantification in various settings where (not) controlling for certain covariates would introduce causal bias.
\end{abstract}

\section{Introduction}

In probabilistic modeling we are often concerned with estimating the posterior predictive distribution $p(Y|X=x,\cD)$ from observed data $\cD$, generated by some unknown data generating mechanism~\cite{Bishop_2006, koller2009probabilistic}. The posterior predictive captures our beliefs about the association of some input variable $X$ and the corresponding outcome variable $Y$ when both were generated by the same mechanism as the training data. However, in many practical applications, new inputs might not be generated by the exact same mechanism; we may face selection bias; or we aim for a fair model explaining outcomes primarily based on cause. In these cases, we may consider a causal model and estimate the counterfactual distribution (posterior interventional distribution), $p(Y|X\leftarrow x, \cD)$~\cite{lauritzen2001causal, pearl2009, pearl2009causality}, which states the posterior probability of an outcome if we intervene in the data generating mechanism and set treatment variable (input)  $X$ to $x$, also written $do(X=x)$.

To illustrate the difference, consider watching people (not) wearing a raincoat and we are interested in predicting street conditions. In a purely predictive setting, we would learn from past data that observing raincoats is correlated with wet streets and  predict street conditions accordingly. However, such a model may fail when applied to a different population of people. For instance, in windy regions, people may wear raincoats even on sunny days; in warm regions, people may prefer umbrellas. In contrast, the posterior interventional distribution would tell us that putting on a coat does not affect street conditions and our prediction would purely rely on the marginal probability of observing wet streets.

In many applications, the situation is more nuanced and we are facing the problem whether to include certain covariates in the model to improve predictive performance or whether to exclude them to not incur into bias. The problem is well-understood if we aim to avoid any causal bias; however, what if we are willing to accept some bias in order to improve predictive performance?

Starting from a given causal probabilistic model defined in terms of structural equations, we are interested in quantifying the causal bias of a predictive model which, together with its posterior probability, allows us to balance predictive performance with causal alignment. We focus on the setting of continuous treatment variables, which will enable us to state an analytic expression of the causal bias.

\subsection{Related Work}

Large parts of literature on causal inference focus on {\it causal identifiability}, that is, identifying conditions under which the causal relationship can be unbiasedly estimated from the observed data~\cite{Spirtes2000, Shpitser2006, reason:Pearl09a}, and {\it causal estimation} of models when these conditions are satisfied. 

If we are not willing to make any additional (parametric) modeling assumptions, causal identifiability requires treatments from a finite domain~\cite{angrist2008mostly, hernan2020causal}. In particular, for binary treatments, we might only be interested in the causal effect expressed in terms of the average treatment effect of the treated (ATET) or the average treatment effect (ATE)~\cite{imbens2015causal}, which are the expected difference between the (potential) outcomes under the two alternative treatments for the treated and the whole population, respectively.

In situations where different posterior interventional distributions are compatible with the observed data and the assumed causal relations, the causal effect is {\it non-identifiable} without making additional modeling assumptions, for instance, on the functional relationship of treatments and outcomes; this is in particular true for continuous treatments~\cite{rosenbaum2010design, NIPS2016_aff16212, hines2021parameterising}. Modeling the treatment assignment mechanism allows for {\it partial identification} of the causal effect through inverse propensity weighting~\cite{imbens2000role, hirano2004propensity, imai2004causal, wu2018matching, austin2019assessing}, modeling the outcome mechanism allows for regression methods~\cite{chipman2010bart, hill2011bayesian}, and doubly robust estimation can be used to combine both~\cite{kennedy2017nonparametric, chernozhukov2018double}. More generally, structural causal models~\cite{reason:Pearl09a} allow us to describe the full data generating mechanism.

Similarly, estimating the average treatment effect for continuous treatments either requires some discretization or assuming a certain functional form called dose-response function in case of bounded univariate treatments~\cite{altshuler1981modeling, gill2001causal, wang2015exposure, galagate2016causal}.

However, most existing works on continuous treatments assume that the causal effect is identifiable and, hence, the causal bias is zero. In contrast, we consider the scenario where the causal effect might be non-identifiable and develop a characterization of the bias that can be used to evaluate how far a model is from being causal. While different in methodology, a similar purpose can be found in~\cite{tran2016model, gain2018structure}. In addition, it can be used to perform sensitivity analysis in the presence of unobserved confounders in a completely non-linear framework, similarly to what is proposed in~\cite{cinelli2020making} for the linear case. Our proposed criterion for identifiability via covariate adjustment is also related to the adjustment criterion in~\cite{shpitser2012validity}, which was shown by the authors to imply and be implied by conditional ignorability~\cite{rosenbaum1983central, little2019statistical}. Our criterion is useful to build the link between identifiability and causal bias being zero.

\subsection{Preliminaries}

We assume a probabilistic model of the data generating mechanism, that is, a joint distribution over some random variables $\cV = \{V_1, \ldots , V_n\}$ and a causal directed acyclic graph (DAG) $\cG$ encoding the direct causal relationship between them~\cite{lauritzen2001causal}. We denote {\it parents}, {\it ancestors}, and {\it descendants} in this graph by $\Pa(V_i) \subset \cV$, $\An(V_i) \subset \cV$, and $\De(V_i) \subset \cV$, respectively.

We further assume that the probabilistic model and causal graph are stated in terms of a canonical structural causal model (SCM)~\cite{reason:Pearl09a} consisting of a set of independent random variables $\cU = \{U_{V_1}, \ldots, U_{V_n}\}$ with some probability densities $p(U_{V_i})$, and a set of functions $\cF = \{f_{V_1}, \ldots, f_{V_n}\}$ such that $V_i := f_{V_i}(\Pa(V_i), U_{V_i}) \ \forall i$. The {\it exogenous variables} $U_{V_i}$ are latent and cannot be observed, whereas the {\it endogenous variables} $V_i$ may or may not be observed.

Technically, an SCM defines a joint probability distribution over $\cU$ and $\cV$ satisfying the global Markov property, where marginalizing over $\cU$ yields a joint {\it observational distribution} over all endogenous variables $\cV$. We denote the graph corresponding to an SCM by $\cG^+$ and denote parents, ancestors, and descendants in $\cG^+$ accordingly by $\Pa^+(\cdot) \subset \cV \cup \cU$, $\An^+(\cdot) \subset \cV \cup \cU$, and $\De^+(\cdot) \subset \cV \cup \cU$. 
Figure~\ref{fig:dag_normal} shows an example of a causal graph $\cG$ and Figure~\ref{fig:dag_reparameterized} shows a graph $\cG^+$ of an SCM that aligns with $\cG$.

\begin{figure}[h!]
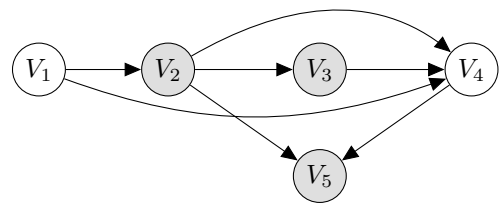
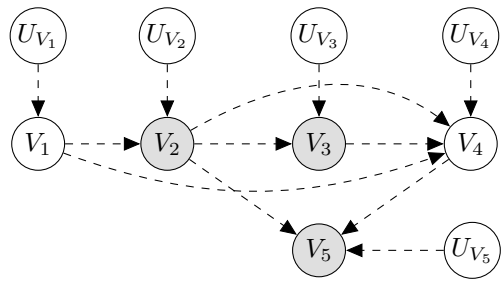

    \vspace{-10pt}
    \begin{subfigure}[b]{0.5\textwidth}
    	\centering
	\tikz{
		\node[latent] (V1) {$V_1$};%
		\node[obs, right=of V1] (V2) {$V_2$};%
		\node[obs, right=of V2, xshift=0.3cm] (V3) {$V_3$}; %
		\node[latent, right=of V3, xshift=0.3cm] (V4) {$V_4$}; %
		\node[obs, below=of V3, yshift=0.3cm] (V5) {$V_5$}; %
		\draw[->] (V1) to (V2);
		\draw[->] (V1) to[out=-20,in=-160] (V4);
		\draw[->] (V2) to (V3);
		\draw[->] (V2) to (V5);
		\draw[->] (V2) to[out=30,in=140] (V4);
		\draw[->] (V3) to (V4);
		\draw[->] (V4) to (V5);
	}
    	\caption{Example of a causal graph $\cG$.}\label{fig:dag_normal}
    \end{subfigure}
    \linebreak
    \begin{subfigure}[b]{0.5\textwidth}
    	\vspace{10pt}
	\centering
    	\tikz{
		\node[latent] (V1) {$V_1$};%
		\node[obs, right=of V1] (V2) {$V_2$};%
		\node[obs, right=of V2, xshift=0.3cm] (V3) {$V_3$}; %
		\node[latent, right=of V3, xshift=0.3cm] (V4) {$V_4$}; %
		\node[obs, below=of V3, yshift=0.3cm] (V5) {$V_5$}; %
		\node[latent,above=of V1, yshift=-0.3cm] (U1) {$U_{V_1}$}; %
		\node[latent,above=of V2, yshift=-0.3cm] (U2) {$U_{V_2}$}; %
		\node[latent,above=of V3, yshift=-0.3cm] (U3) {$U_{V_3}$}; %
		\node[latent,above=of V4, yshift=-0.3cm] (U4) {$U_{V_4}$}; %
		\node[latent,right=of V5, xshift=0.3cm] (U5) {$U_{V_5}$}; %
		\draw[dashed, ->] (V1) to (V2);
		\draw[dashed, ->] (V1) to[out=-20,in=-160] (V4);
		\draw[dashed, ->] (V2) to (V3);
		\draw[dashed, ->] (V2) to (V5);
		\draw[dashed, ->] (V2) to[out=30,in=140] (V4);
		\draw[dashed, ->] (V3) to (V4);
		\draw[dashed, ->] (V4) to (V5);
		\draw[dashed, ->] (U1) to (V1);
		\draw[dashed, ->] (U2) to (V2);
		\draw[dashed, ->] (U3) to (V3);
		\draw[dashed, ->] (U4) to (V4);
		\draw[dashed, ->] (U5) to (V5);
	}
	\caption{Example of reparameterized causal graph $\cG^+$.}\label{fig:dag_reparameterized}
    \end{subfigure}
    \vspace{-10pt}
    \caption{Example of a causal graph (a) and a (reparameterized) causal graph corresponding to an SCM that agrees with $\cG$. Grey and white circles denote observed and unobserved variables, respectively. Stochastic dependencies are encoded through solid arrows whereas dashed arrows highlight that the transitions are deterministic given all parents.}
\end{figure}

Among the endogenous variables, we assume a single, potentially multivariate, treatment variable $X \in \cV$ with $\cX := \{X\}$, and a single outcome variable $Y \in \cV$ with $\cY := \{Y\}$. We further denote the remaining set of {\it observed} endogenous variables by $\cO \subseteq \cV \setminus (\cX \cup \cY)$, while all other endogenous variables $\cL := \cV \setminus (\cX \cup \cY \cup \cO)$ are {\it latent}. Among the observed variables $\cO$, we denote non-descendants of $X$ by $\underline{\cO} :=\cO \setminus \De(X)$, similarly to the notation in mutilated graphs~\cite{pearl2009causality}. We denote the corresponding set of exogenous variables by $\cU_{\cO}$, $\cU_{\underline{\cO}}$, and $\cU_{\cL}$, respectively.

We use lowercase characters to denote realizations of random variables, for instance, $v_i$ for a realization of $V_i$, and we may use $p(\cdot | v_i)$ as a short form for the conditional probability density (or mass) $p(\cdot | V_i=v_i)$. We use capital letters without indices to denote concatenated vectors of random variables, for instance, $O := [O_j]_{j=1}^m$ with $\cO := \{O_1,\ldots,O_m\}$, and $U_{\cO} := [U_{O_j}]_{j=1}^m$. Likewise, realizations without indices refer to vectors of realizations like $o := [o_j]_{j=1}^m$, and $f_O$ denotes the concatenation of all $f_{O_j}$, which is a vector-valued function with combined set of arguments.

For a generic function $h$ and argument(s) $x$, the Jacobian of $h$ with respect to $x$ is denoted by $\nabla_{x}h$ if it exists. Note that if $h$ is a scalar function, $\nabla_{x}h$ denotes a transposed gradient. We use $\bE$, $\Cov$ and $\Var$ for expectation, covariance, and variance. We may use subscripts to specify the underlying probability density, e.g.~$\bE_{Y|x,o}$ means that the expectation is taken with respect to $p(Y|x,o)$.

\section{Identification, quantification and estimation}\label{sec:identification_quantification_inference}

In this section we (i) introduce a novel criterion for identifiability of causal effects via covariate adjusts; (ii) we derive closed-form expressions for causal effect and causal bias that can be computed in closed-form; and (iii) show how the latter can be estimated via standard probabilistic inference methods.

\subsection{A characterization of identifiability via covariate adjustment}\label{sec:identification}

This section offers a complete criterion for identifiability of causal effects~\cite{pearl2009causality} via covariate adjustment~\cite{shpitser2012validity} in the framework of structural causal models. 
The underlying idea is to first construct a minimal set of exogenous variables, $\cU^C \subseteq \cU$, such that, together with variables in $\cX$ and $\underline{\cO}$, this set fully characterizes the randomness of outcome variable $Y$.
To this end, let us again consider the causal graph in Figure~\ref{fig:dag_normal} as a running example.
\begin{example}\label{ex:running_example}
We assume the causal graph in Figure~\ref{fig:dag_normal} and $X = V_2$, $Y = V_4$, $\cO = \{V_3, V_5\}$, and $\cL = \{V_1\}$. In this case, it is well understood that the unobserved variable $V_1$ and the observed variables $V_3$ and $V_5$ each compromise identifiability of the causal effect from $X$  to $Y$. Latent variable $V_1$ is a hidden {\it confounder} introducing confounding bias, observing {\it mediator} $V_3$ introduces overcontrol bias, and observing {\it collider} $V_5$ introduces (endogenous) selection bias~\cite{elwert2013graphical}. Figure~\ref{fig:dag_reparameterized} is the causal graph $\cG^+$ of an SCM aligning with $\cG$. We notice that \[  U_{V_1}\not\indep X, \quad U_{V_3}\not\indep X |V_3,\quad U_Y\not\indep X | V_5, \] which we will see in Theorem~\ref{thm:causal_identification} to be the very reasons why the causal effect is not identifiable through covariate adjustment.
\end{example}

\begin{definition}\label{def:causal_set}
For any set of treatment variables $\cX$ and outcome variable $Y$, the set of exogenous random variables which are ancestors of $Y$ in $\cG^+$ and that are not independent of $Y$ given $\cX$ and $\underline{\cO}$ is called the \textit{causal exogenous set}:
\begin{equation}\label{eq:causal_set}
        \cU^C :=\{U_{V_i}\in \cU\cap \An^+(Y): U_{V_i}\not\indep Y | \cX, \underline{\cO}\}.
\end{equation} 
\end{definition}

Intuitively, this set contains all exogenous variables which we would need to control for to translate the observational setting into an experimental setting. In Example~\ref{ex:running_example} this set includes random variables $U_{V_1}$, $U_{V_3}$, and $U_{Y}$. While we can never observe exogenous variables, we might be able to obtain a posterior distribution over them by making (weak) prior distributional assumptions. If this posterior distribution is insensitive to the treatment variable, the causal effect can be (partially) identified. In the following, we state this observation more formally.

\vspace{5pt}
\begin{assumption}\label{as:positive}
We assume that the given SCM is defined through continuous exogenous variables $U_{V_i}$ where the probability density exists and is strictly positive for all possible values of $U_{V_i}$.
\end{assumption}

\vspace{5pt}
\begin{assumption}\label{as:faithful}
We assume that the given SCM is \textit{faithful} to the causal graph $\cG^+$, that is, any independences that hold in the model must hold in the causal graph.
\end{assumption}

\vspace{5pt}
\begin{theorem}\label{thm:causal_identification}
Given an SCM that satisfies Assumption~\ref{as:positive} with causal exogenous set $\cU^C$ for treatment variable $X$, outcome variable $Y$, and a set of observations $\cO$, the causal effect of $X$ on $Y$ can be identified by covariate adjustment if, and only if,
	\begin{equation}\label{eq:causal_identification}
	U_{V_i} \indep X | \cO \ \ \ \forall \ U_{V_i} \in \cU^C
	\end{equation}
\end{theorem}
\begin{proof}
    See Appendix \ref{app:causal_identification}.
\end{proof}

\begin{corollary}\label{cor:grad_p_zero}
Given an SCM that satisfies Assumptions~\ref{as:positive}-\ref{as:faithful} where the causal effect of $X$ on $Y$ can be identified by covariate adjustment and which is differentiable in $X$, then
\begin{equation}\label{eq:grad_p_zero}
    \left. \nabla_X \log p(U^C | X,O)\right|_{X=x} = 0 \ \ \ \forall \ x,
\end{equation}
where $U^C$ denotes the vector of all endogenous variables in $\cU^C$ and $O$ denotes the vector of observations.
\end{corollary}
\begin{proof}
According to Theorem~\ref{thm:causal_identification}, treatment variable $X$ is independent of all elements of the causal exogenous set given the observed variables, $p(U^C|X,O)=p(U^C|O)$. If the density defined through the SCM exists and is differentiable w.r.t. $X$, the partial gradient $\nabla_X \log p(U^C | X,O)$ is zero everywhere.
\end{proof}

\subsection{Quantification of average partial effect and bias}\label{sec:quantification}

Corollary~\ref{cor:grad_p_zero} gives rise to a measure of causal bias which is zero if the causal bias is zero. To formalize this, let us introduce the notation of {\it partial evaluation}, which is related to the potential outcome notation and interventions in structural causal models.

\vspace{5pt}
\begin{definition}
Given a function $f$ with arguments $\cA = \{V_{i_1},\ldots,V_{i_m}\} \subseteq \cV$ and some set $\cB = \{V_{j_1},\ldots,V_{j_k}\} \subseteq \cV$, we define $f^b$ to be the {\it partial evaluation} of $f$ where all arguments $\cA \cap \cB$ have been assigned to the corresponding value $b=[b_1,\ldots,b_k]$, that is
    \begin{equation}
        f^{\cB\leftarrow b}(\cA \setminus \cB) := f(V_{i_1},\dots,V_{i_m})\big|_{V_{j_1}=b_1,\ldots,V_{j_k}=b_k}.
    \end{equation}
\end{definition}

We may use the short form $f^b$ for $f^{\cB\leftarrow b}(\cA \setminus \cB)$ if the corresponding set of assigned random variables $\cB$ is clear from the context, for example, all observed variables $\cO$. Likewise, if sets $\cB$ and $\cC$ are disjoint sets of random variables, we use $f^{b,c}$ to denote simultaneously assignments $\cB \leftarrow b$ and $\cC \leftarrow c$.

\vspace{5pt}
\begin{example}
Consider again the causal graph in Figure~\ref{fig:dag_reparameterized} with $X = V_2$ with realization $x$, $\cO = \{V_3, V_5\}$ with realization $o = [o_1, o_2]$, $Y = V_4$, and $\cL = \{V_1\}$. Here, for instance, we have
\begin{eqnarray*}
f_Y^x &=& f_Y(V_1, x, f_{V_3}(x, U_{V_3}), U_Y),\\
f_{V_5}^{x,o} &=& f_{V_5}(x, f_Y(V_1, x, o_1, U_Y), U_{V_5}).
\end{eqnarray*}
If derivatives are well-defined, this for example implies $\nabla_{U_{V_3}} f_{V_5}^{x,o}=0$, while in general $\nabla_{U_{V_3}} f_{V_3} \ne 0$.
\end{example}

\vspace{5pt}
\begin{proposition}\label{thm:marginal_decomposition}
Given an SCM that satisfies Assumptions~\ref{as:positive}-\ref{as:faithful} with a vector of causal exogenous variables $U^C$ for treatment variable $X$, outcome variable $Y$, and observations $O$. Suppose that $f_Y$ and $p$ are differentiable with respect to $X$ and satisfy sufficient regularity assumptions. Then 
\begin{eqnarray}
  \underbrace{\nabla_x\bE_{Y|x,o}[Y]}_{\cA(x, o)} &\!\!=\!\!& \underbrace{\bE_{U^C|x,o}\left[\nabla_x f_Y^{x,\underline{o}}\right]}_{\cC(x, o)} \ + \nonumber\\ 
  && \underbrace{\bE_{U^C|x,o}\left[f_Y^{x,\underline{o}} \nabla_x\! \log p(U^C|x, o)\right]}_{\cB(x, o)}\!. \ \ \ \ \label{eq:marginal_decomposition}
\end{eqnarray}
If the causal effect of $X$ on $Y$ can be identified by covariate adjustment then $\cB(x, o)=0$.
\end{proposition}
\begin{proof}
	See Appendix \ref{app:marginal_decomposition}.
\end{proof}

In the above proposition, expression $\cA(x, o)$ states the partial derivative of the average observational outcome, $\cC(x, o)$ states average partial effect on the treated (APET) for an infinitesimal change in the treatment, and $\cB(x, o)$ is the causal bias.

\begin{remark}
The decomposition in~\eqref{eq:marginal_decomposition} has a well-known discrete counterpart for some reference treatment $x$ (control) and proposal $x' := x + h$ (treated),
\begin{eqnarray*}
\underbrace{\bE[Y|x'] - \bE[Y|x]}_{\text{association}} &\!\!=\!\!& \underbrace{\bE[Y^{x'}-Y^x|x']}_{\text{ATET}} \ + \\
&& \underbrace{\bE[Y^x|x'] - \bE[Y^x|x]}_{\text{bias}} \, ,\ 
\end{eqnarray*}
where association is defined as the contrast in average outcomes by observed treatment value, and ATET refers to the average treatment effect on the treated. If we condition on $O=o$, divide both sides by $h$ and take the limit for $h\to 0$, we recover~\eqref{eq:marginal_decomposition}. 
\end{remark}

The expressions for $\cC(x,o)$ and $\cB(x,o)$ in~\eqref{eq:marginal_decomposition} highlight the connection with Theorem~\ref{thm:causal_identification}. However, in practice we do not need to find the causal exogenous set $\cU^C$ in order to compute them. Indeed, because $\nabla_x f_Y^{x,\underline{o}}$ does not depend on $\cU\setminus \cU^C$ by definition of $\cU^C$, via marginalization we can write
\begin{equation}\label{eq:marginal_causal_effect}
    \cC(x,o)=\bE_{U|x,o}\left[\nabla_x f_Y^x\right],
\end{equation}
where $U$ denotes the vector of all exogenous variables and we are free to remove the assignment of $\underline{o}$. Regarding $\cB$, the main challenge of the expression in~\eqref{eq:marginal_decomposition} is to compute the term $\nabla_x \log p(U^C|x,o)$, since $x$ appears in the normalization constant. Theorem~\ref{thm:bias_as_covariance} provides an elegant formulation of the bias that is easier to compute in practice. Under some additional assumptions Theorem~\ref{thm:bias_rewrite} provides an alternative formulation of $\cB$ that is automatically differentiable and does not depend on $U^C$. 

\begin{theorem}\label{thm:bias_as_covariance}
Under the assumptions of Proposition \ref{thm:marginal_decomposition}, the bias in \eqref{eq:marginal_decomposition} can be rewritten as 
\begin{equation}\label{eq:bias_as_covariance}
    \cB(x,o) = \Cov_{U^C|x,o}\left(f_Y^{x,\underline{o}}, \nabla_x \log p(U^C, x, o)\right),
\end{equation}
where $p(U^C, x, o)$ denotes the joint density over $U^C$, $X$, and $O$ (partially) evaluated for $x$ and $o$.
\end{theorem}
\begin{proof}
  See Appendix \ref{app:bias_as_covariance}.
\end{proof}

The expression of the bias in Theorem~\ref{thm:bias_as_covariance} is a clear improvement over~\eqref{eq:marginal_decomposition} as the joint $p(U^C,X,O)$ does not include a normalization constant involving $X$. If no mediators between $X$ and $Y$ are observed (e.g.~$V_3$ in Example~\ref{ex:running_example}),~\eqref{eq:bias_as_covariance} can already provide a practical way to estimate the bias. Otherwise it may hide a difficulty: If the factorization of the joint includes a Dirac delta distribution depending on $X$, computing its gradient with respect to $X$ requires some extra analysis.

For instance, consider Example~\ref{ex:running_example} where the joint contains conditional $p(V_3|X,U_{V_3})$ as $V_3 \in \cO$. Since this conditional is a Dirac delta distribution, how can we then compute the gradient of $\log p(U^C,X,O)$ w.r.t. $X$? The following theorem further manipulates the bias to remove this obstacle.

\begin{theorem}\label{thm:bias_rewrite}
Suppose that $f_Y$ is differentiable in $X$ and all $V_i \in \cO$, and that corresponding functions $f_X$ and $f_{V_i}$ are differentiable and invertible in $U_X$ and $U_{V_i}$, respectively. Then the bias in~\eqref{eq:marginal_decomposition} can be rewritten as
    \begingroup\makeatletter\def\f@size{8}\check@mathfonts
    \def\maketag@@@#1{\hbox{\m@th\normalsize\normalfont#1}}
    \begin{align}\label{eq:bias_rewrite}
    \cB(x,o) \ = \!\!
    \sum_{V_i\in \cX\cup \cO}\!\!\!\!\!\!\bE_{U|x,o}\Big[&\Big(\nabla_{U_{V_i}} f_Y^x \!+\! \left(f_Y^x-y^x\right)\nabla_{U_{V_i}}\!\log p(U_{V_i})\notag \Big)\\&\ (\nabla_{U_{V_i}}f_{V_i}^{x,o})^{-1}\, \nabla_x (f_{V_i}^{x,o}  - v_i)\Big]
    \end{align}
    \endgroup
with potential outcome $y^x := \bE_{U|x,o}[f_Y^x]$ for treatment $x$.
\end{theorem}
\begin{proof}
    See Appendix \ref{app:bias_rewrite}.
\end{proof}

\begin{remark}
Note that (i) all terms in the expectation of~\eqref{eq:bias_rewrite} can be efficiently computed via automatic differentiation. (ii) Bias $\cB(x,o)$ can be decomposed in the contributions given by $X$ (potential confounding bias) and by each variable $V_i \in \cO$ (potential overcontrol or selection bias). 
(iii) If $U_{V_i}$ is modeled as a (multivariate) standard Gaussian, then $\nabla_{U_{V_i}}\log p(U_{V_i}) = -U_{V_i}$. Finally, (iv) the inverse Jacobian $(\nabla_{U_{V_i}}f_{V_i}^{x,o})^{-1}$ can be efficiently computed if $V_i$ is low-dimensional or if $f_{V_i}$ has, for example, a normalizing flow structure~\cite{papamakarios2019normalizing}. 
\end{remark}

\subsection{Estimation of average partial effect and bias}\label{sec:estimation}

Sections~\ref{sec:identification} and \ref{sec:quantification} provide closed-form expressions for average partial effect on the treated and the causal bias, namely~\eqref{eq:marginal_causal_effect} and~\eqref{eq:bias_rewrite}. Their numerical estimation reduces to a mere statistical inference task: the approximation of an expectation with respect to the probability density $p(U|x,o)$. We further notice that
\begin{equation}\label{eq:posterior_split}
    p(U|x,o) = p(U_{\cL}|x,o) \prod_{V_i \in \cX \cup \cO}\delta(U_{V_i} - u_{V_i}^*),
\end{equation} 
where $u_{V_i}^*$ are uniquely determined by the assumptions in Theorem~\ref{thm:bias_rewrite} that $f_{V_i}$ are invertible in $U_{V_i}$ for all $V_i \in \cX \cup \cO$ given $U_{\cL}$. Hence, we can restrict our attention to infer $p(U_{\cL}|x,o)$. In some cases, such a posterior might be available in closed-form. Otherwise we may resort to sampling methods such as MCMC~\cite{brooks2011handbook}, SMC~\cite{doucet2001introduction}, and importance sampling~\cite{neal2001annealed}, or variational methods like Laplace approximation~\cite{shun1995laplace}, variational Bayes~\cite{Jordan:1998:vi, blei2017variational}, and normalizing flows~\cite{rezende2015variational, papamakarios2019normalizing}. Let us quickly review Laplace approximation and importance sampling in our framework, as we will use them in Sections~\ref{sec:lesser_evil} and~\ref{sec:ascvd}, respectively.
 
\textbf{Laplace approximation.} The Laplace approximation is a Gaussian approximation of $p(U_{\cL}|x,o)$ around its mode. In order to compute the mode, we first notice that via automatic differentiation we have immediate access to $\nabla_{U_{V_i}}\log p(U_{\cL},x,o)$ for all $U_{V_i} \in \cU_{\cL}$, which we can employ in any gradient-based optimization to compute a maximum-a-posteriori (MAP) estimate~\cite{murphy2012machine}. Let us denote the latter by $u_{\cL}^{\MAP}$. Furthermore, automatic differentiation also gives access to the Hessian with elements (sub-matrices) $\nabla_{U_{V_i}}\nabla_{U_{V_j}}\log p(U_{\cL},x,o)$ or to any scalable approximation of it (e.g.~diagonal approximation). Denote this Hessian evaluated at the MAP estimate by $H(u_{\cL}^{\MAP})$. Whenever $p(U_{\cL}|x,o)$ has a well-defined MAP estimate,~$-H(u_{\cL}^{\MAP})$ is positive-definite and we can approximate the posterior by a (multivariate) Gaussian,
 \begin{equation}\label{eq:laplace_approximation}
 p(U_{\cL}|x,o)\approx \cN\!\left(u_{\cL}^{\MAP}, -H^{-1}(u_{\cL}^{\MAP})\right)\!(U_{\cL}).
 \end{equation}
We can now generate a sample $u_{\cL}^{(i)}$ from the approximated posterior, compute the remaining exogenous variables via the respective inverse maps of $f_{V_i}^{x,o}$, and produce a Monte Carlo estimate of $\cC(x,o)$ and $\cB(x,o)$ using~\eqref{eq:marginal_causal_effect} and~\eqref{eq:bias_rewrite}, respectively.
 
\textbf{Importance sampling.} Importance sampling is a way to modify the density underlying an expectation. In causal inference, it is used for inverse propensity weighting (IPW)~\cite{imbens2000role, hirano2004propensity, imai2004causal, wu2018matching, austin2019assessing}. Consider any treatable density $q(U_{\cL})>0$ of a sampling distribution, and some function $h(u_{\cL}, u_X, u_{\cO})$, which is integrable in all arguments. We have
 \begin{eqnarray}
     \bE_{U|x,o}[h(U_{\cL},\!\!\!\!\!\!\!\!&&\!\!\!\!\!\!\!\! U_X, U_{\cO})] \nonumber \\
     &=&\bE_{U_{\cL}|x,o}\left[h(U_{\cL}, u_X^*, u_{\cO}^*)\right] \nonumber\\
     &=&\bE_{U_{\cL}}\left[h(U_{\cL}, u_X^*, u_{\cO}^*)w(U_{\cL}|x,o)\right], \ \ \ \ \ \ \label{eq:importance_sampling}
 \end{eqnarray}
where $\bE_{U_{\cL}}$ denotes the expectation with respect to density $q(\cU_{\cL})$ and normalized importance weights
\begin{equation} \label{eq:importance_weight}
w(U_{\cL}|x,o) \propto \frac{p(U_{\cL},x,o)}{q(U_{\cL})},
\end{equation}
with $\bE_{U_{\cL}}\left[w(U_{\cL}|x,o)\right]=1$. Then we can first sample $u_{\cL}^{(i)}$ from $q(U_{\cL})$ for $i=1,\ldots,N$, compute the unnormalized weights $\tilde{w}^{(i)} = p(u_{\cL}^{(i)},x,o) \, /\,q(u_{\cL}^{(i)})$, and normalize the importance weights to sum up to $N$. Finally, using these same samples, we can compose a Monte Carlo estimator using~\eqref{eq:importance_sampling}, and estimate $\cC(x,o)$ and $\cB(x,o)$ by applying this technique for $h$ being the respective expressions in~\eqref{eq:marginal_causal_effect} and~\eqref{eq:bias_rewrite}. Notice that if we take $q(U_{\cL})$ to be the prior $p(U_{\cL})$, the unnormalized weights simplify to be the likelihood $p(x,o|u_{\cL}^{(i)})$.

\section{Experiments}\label{sec:experiments}

\begin{figure*}[h!]
\noindent\begin{minipage}{.3\textwidth}
	\scalebox{0.9}{
	\tikz{
		\node[latent] (V1) {$V_1$};%
		\node[obs, right=of V1] (X) {$X$};%
		\node[latent, right=of X, xshift=0.3cm] (Y) {$Y$}; %
		\node[latent,above=of V1, yshift=-0.3cm] (U1) {$U_{V_1}$}; %
		\node[latent,above=of X, yshift=-0.3cm] (UX) {$U_X$}; %
		\node[latent, above=of Y, yshift=-0.3cm] (UY) {$U_Y$}; %
		\draw[dashed, ->] (V1) to[out=30,in=140] (Y);
		\draw[dashed, ->] (X) to (Y);
		\draw[dashed, ->] (V1) to (X);
		\draw[dashed, ->] (U1) to (V1);
		\draw[dashed, ->] (UX) to (X);
		\draw[dashed, ->] (UY) to (Y);
	}
	}
\begin{align*}
    V_1 &= U_{V_1}\\
    X &= \alpha V_1 + U_X\\
    Y &= \beta X + \gamma V_1 + U_Y
\end{align*}
\caption{Confounding bias}\label{fig:confounding_model} 
\end{minipage}
\hspace{0.3cm}
\noindent\begin{minipage}{.3\textwidth}
	\scalebox{0.9}{
	\tikz{
		\node[obs] (X) {$X$};%
		\node[obs, right=of X, xshift=0.3cm] (V1) {$V_1$}; %
		\node[latent, right=of V1, xshift=0.3cm] (Y) {$Y$}; %
		\node[latent,above=of X, yshift=-0.3cm] (UX) {$U_X$}; %
		\node[latent,above=of V1, yshift=-0.3cm] (U1) {$U_{V_1}$}; %
		\node[latent, above=of Y, yshift=-0.3cm] (UY) {$U_Y$}; %
		\draw[dashed, ->] (X) to[out=30,in=140] (Y);
		\draw[dashed, ->] (X) to (V1);
		\draw[dashed, ->] (V1) to (Y);
		\draw[dashed, ->] (UX) to (X);
		\draw[dashed, ->] (U1) to (V1);
		\draw[dashed, ->] (UY) to (Y);
	}	
	}
\begin{align*}
    X &= U_X\\
    V_1 &= \alpha X + U_{V_1}\\
    Y &= \beta X + \gamma V_1 + U_Y
\end{align*}
\caption{Overcontrol bias}\label{fig:overcontrol_model} 
\end{minipage}
\hspace{0.5cm}
\noindent\begin{minipage}{.3\textwidth}
	\scalebox{0.9}{
	\tikz{
		\node[obs] (X) {$X$};%
		\node[latent, right=of X, xshift=0.3cm] (Y) {$Y$}; %
		\node[obs, right=of Y, xshift=0.3cm] (V1) {$V_1$}; %
		\node[latent,above=of V1, yshift=-0.3cm] (U1) {$U_{V_1}$}; %
		\node[latent,above=of X, yshift=-0.3cm] (UX) {$U_X$}; %
		\node[latent, above=of Y, yshift=-0.3cm] (UY) {$U_Y$}; %
		\draw[dashed, ->] (X) to (Y);
		\draw[dashed, ->] (X) to[out=30,in=140] (V1);
		\draw[dashed, ->] (Y) to (V1);
		\draw[dashed, ->] (U1) to (V1);
		\draw[dashed, ->] (UX) to (X);
		\draw[dashed, ->] (UY) to (Y);
		\draw[dashed, ->] (U1) to (V1);
	}	
	}
\begin{align*}
    X &= U_X\\
    Y & = \alpha X + U_Y\\
    V_1 &= \beta X + \gamma Y + U_{V_1}
\end{align*}
\caption{End.~selection bias}\label{fig:selection_model}
\end{minipage}
\end{figure*}

In this section we show several experiments that validate the theory developed in Section~\ref{sec:identification_quantification_inference}, and illustrate interesting applications. In Section~\ref{sec:lin_biases}, we provide closed-form expressions for all possible graphical structure inducing causal bias, and linear structure equation models. in Section~\ref{sec:lesser_evil}, we show how quantifying the bias can be helpful in situation where full identifiability cannot be achieved due to missing data. Finally, in Section~\ref{sec:ascvd} we study the application of our framework to a simulated medical study.

\subsection{Confounding, overcontrol and selection bias}\label{sec:lin_biases}

First, we study simple linear models exhibiting the three possible types of bias: confounding, overcontrol and (endogenous) selection~\cite{elwert2013graphical}. Studying linear models is useful because we can compute $\cA$, $\cC$ and $\cB$ in closed-form, checking our results and building up intuition. For all models, we will consider $\alpha,\beta,\gamma$ and $\delta$ to be scalar parameters and $U_{V_1}, U_X, U_Y\sim \cN(0,1)$. For detailed proofs of the following results, see Appendix \ref{app:bias_types}.

\textbf{Confounding bias.}
Given the model in Figure \ref{fig:confounding_model}, we can compute 
    \begingroup\makeatletter\def\f@size{9.5}\check@mathfonts
\def\maketag@@@#1{\hbox{\m@th\normalsize\normalfont#1}}
\[
\cA(x,o) \, = \, \beta + \frac{\gamma\alpha}{1+\alpha^2}, \quad
\cC(x,o) \, = \, \beta, \quad
\cB(x,o) \, = \, \frac{\gamma\alpha}{1+\alpha^2}.
\]
\endgroup
Notice that because the model is linear, $\cC(x,o)$ is independent of $p(U|x,o)$. If $\gamma\ne 0$ then $U_{V_1}\in \cU^C$, and if $\alpha\ne 0$ then $U_{V_1}\not\indep X$. Thus \eqref{eq:causal_identification} is violated and $\cB(x,o)\ne 0$.

\textbf{Overcontrol bias.}
Given the model in Figure \ref{fig:overcontrol_model}, we can compute 
    \begingroup\makeatletter\def\f@size{9.5}\check@mathfonts
\def\maketag@@@#1{\hbox{\m@th\normalsize\normalfont#1}}
\[
\cA(x,o) \, = \, \beta,\quad
\cC(x,o) \, = \, \beta + \gamma\alpha, \quad
\cB(x,o) \, = \, -\gamma\alpha.
\]
\endgroup
Again, because the model is linear, $\cC(x,o)$ is independent of $p(U|x,o)$. If $\gamma\ne 0$ then $U_{V_1}\in \cU^C$, and if $\alpha\ne 0$ then $U_{V_1}\not\indep X|V_1$. Thus \eqref{eq:causal_identification} is violated and $\cB(x,o)\ne 0$.

\textbf{Selection bias.}
Given the model in Figure \ref{fig:selection_model}, we can compute 
    \begingroup\makeatletter\def\f@size{8.5}\check@mathfonts
\def\maketag@@@#1{\hbox{\m@th\normalsize\normalfont#1}}
\[
\cA(x,o) \, = \, \frac{\alpha - \gamma\beta}{1 + \gamma^2}, \quad
\cC(x,o) \, = \, \alpha, \quad
\cB(x,o) \, = \, -\frac{\gamma(\beta+\gamma\alpha)}{1+\gamma^2}.
\]
\endgroup
Once more, because the model is linear, $\cC(x,o)$ is independent of $p(U|x,o)$. If $\gamma\ne 0$ and either $\alpha\ne 0$ or $\beta\ne 0$, then $U_Y\not\indep X|V_1$. Since $U_Y\in \cU^C$, this violates \eqref{eq:causal_identification} and $\cB(x,o)\ne 0$.

\subsection{The lesser of two evils: covariate adjustment with missing data}\label{sec:lesser_evil}
In most real applications we cannot adjust for all confounders. For example, take the causal graph in Figure~\ref{fig:lesser_evil}. Whenever we model a mediator $V_2$, there is likely to exist some variable $V_1$ for which we do not have data acting as a confounder between $X$ and $V_2$. Theorem \ref{thm:causal_identification} says that in order to achieve identifiability we would need to observe $V_1$ and not observe $V_2$. However, as we cannot observe $V_1$, we are left with the following question: Should we observe $V_2$ and incur overcontrol bias, or not observe it and have confounding bias?

\vspace{5pt}
\begin{figure}[h!]
\noindent\begin{minipage}{.5\textwidth}
		\begin{center}
\tikz{
	\node[latent] (V1) {$V_1$};%
	\node[obs, right=of V1] (X) {$X$};%
	\node[latent, right=of X, xshift=0.3cm] (V2) {$V_2$}; %
	\node[latent, right=of V2, xshift=0.3cm] (Y) {$Y$}; %
	\node[latent,above=of V1, yshift=-0.3cm] (U1) {$U_{V_1}$}; %
	\node[latent,above=of X, yshift=-0.3cm] (UX) {$U_X$}; %
	\node[latent,above=of V2, yshift=-0.3cm] (UV2) {$U_{V_2}$}; %
	\node[latent, above=of Y, yshift=-0.3cm] (UY) {$U_Y$}; %
	\draw[dashed, ->] (V1) to[out=30,in=140] (V2);
	\draw[dashed, ->] (X) to (V2);
	\draw[dashed, ->] (V2) to (Y);
	\draw[dashed, ->] (V1) to (X);
	\draw[dashed, ->] (U1) to (V1);
	\draw[dashed, ->] (UX) to (X);
	\draw[dashed, ->] (UV2) to (V2);
	\draw[dashed, ->] (UY) to (Y);
}	
\end{center}
\end{minipage}
\begin{minipage}{.45\textwidth}
\begin{align*}
V_1 &= U_{V_1}\\
X &= \alpha \exp(V_1) + U_X\\
V_2 &= \beta X + \gamma V_1^2 + U_{V_2}\\
Y &= \delta V_2 + U_Y
\end{align*}
\end{minipage}
\caption{Confounding vs.~overcontrol bias. If $V_1$ cannot be observed, should we observe $V_2$ or not?}\label{fig:lesser_evil}
\end{figure}
\vspace{2pt}

Let us bring this question to the non-linear model in Figure~\ref{fig:lesser_evil}, where we take $\alpha,\beta,\gamma,\delta$ to be scalar parameters and $U_{V_1},U_X,U_{V_2},U_Y\sim \cN(0,1)$. Here, the average partial effect on the treated is given by $\cC(x,o)=\beta\delta$ independently whether we observe $V_2$ or not, therefore it does not help us towards a decision. Rather, we want to compute the absolute marginal causal bias $|\cB(x,o)|$ and pick the case where it is smaller. To this purpose, here we first estimate $p(U_{\cL}|x,o)$ via a Laplace approximation as described in Section~\ref{sec:estimation}, using an ADAM optimizer~\cite{kingma2014adam} to compute the MAP; then we compose $p(U|x,o)$ as in~\eqref{eq:posterior_split}. 

\begin{figure*}[h!]
    \centering
    \includegraphics[scale=0.41]{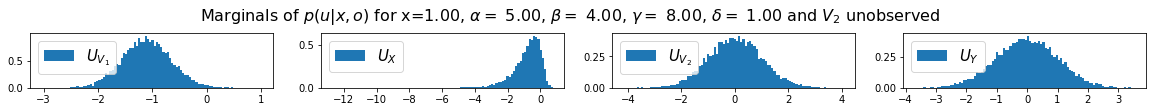}
    \includegraphics[scale=0.41]{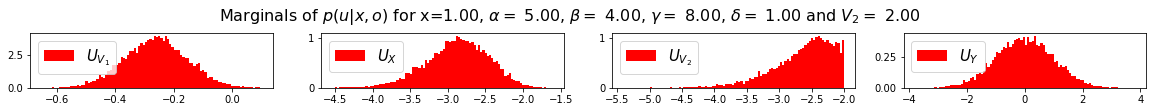}
    \caption{Marginal distributions of $p(U|x,o)$}
    \label{fig:posteriors}
\end{figure*}

\begin{figure*}[h!]
    \centering
    \includegraphics[scale=0.44]{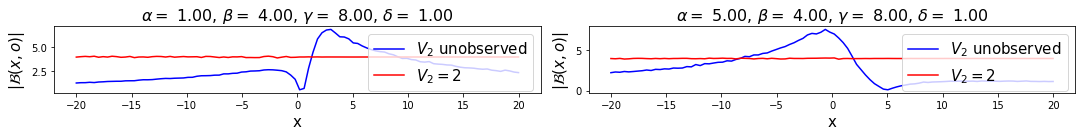}
    \caption{Comparisons of $|\cB(x,o)|$ over $x$, for arbitrary parameters $\alpha,\beta,\gamma,\delta$}
    \label{fig:bias_over_treatment}
\end{figure*}

Figure~\ref{fig:posteriors} shows its marginal distributions for an arbitrary configuration of parameters and for $x=1$, when $V_2$ is unobserved and when $V_2=2$.

Given approximate posterior samples, we can compute a Monte Carlo estimator of the bias in \eqref{eq:bias_rewrite}. Figure~\ref{fig:bias_over_treatment} shows the estimated $|\cB(x,o)|$ both when $V_2$ is unobserved and when $V_2=2$, for two different configurations of parameters $\alpha,\beta,\gamma,\delta$ and for $x\in(-20,20)$. We can see that the estimated $|\cB(x,o)|$ is constant over $x$ when $V_2$ is observed, which is expected since for this model the overcontrol bias equals $\beta\delta$. The confounding bias, however, is non-linear in $x$. Just increasing $\alpha$ from $1$ to $5$, the range of values of $x$ for which observing $V_2$ is better than not observing it changes drastically. As a sanity check, we also make sure that the estimated value of $\cB(x,o)$ is close to $0$ if we could hypothetically observe only $X$ and $V_1$; this turns out in the order of $10^{-13}$, confirming that the estimation procedure is sensible.

In conclusion, as expected the decision depends on specific observations and parameters, potentially learned from data. If we can model missing confounders, feasible bias estimation provides a quantitative tool to take concrete decisions under missing data.

\subsection{A simulated study of statins and atherosclerotic cardiovascular disease}\label{sec:ascvd}
We follow the example in~\cite{zivich2021machine} and conduct a simulation study of statins and subsequent atherosclerotic cardiovascular disease (ASCVD). We adapt the data generation procedure as presented by the authors to satisfy the requirements of Theorem~\ref{thm:bias_rewrite}. In particular, we assume a continuous treatment variable $X$ representing the strength of each prescribed medicine with support in the range $[0, 1]$. This dose intensity of statin is relative to the maximal dose, that is, 80 mg for atorvastatin and 40 mg for rosuvastatin. This better reflects reality as a typical daily dose varies depending on the patient characteristics and other confounders: age $A$, pre-treatment low-density lipoprotein $L$, frailty $F$, diabetes $D$, and ASCVD risk score $R$. We consider another measurement $M$ of low-density lipoprotein that happens post-treatment. The outcome $Y$ indicates the observed incidence of ASCVD. Finally, we record whether a patient had severe headache $H$, which can be an adverse reaction to statin and/or caused by ASCVD. 

We implement the data generation process as a $\theta$-parameterized joint distribution of $X$ and $Y$ as well as covariates $\cV = \{A, L, F, R, D, M, H\}$. The model in use is complex and non-linear; see Appendix~\ref{app:ASCVD}. Figure~\ref{fig:statsmedical} shows some distributions and statistics recovered from the data, which, although the model is not exactly the same, are relatively close to those reported by the authors.

\begin{figure*}[h!]
\centering
\includegraphics[clip, trim=0cm 0.2cm 0cm 0.9cm, scale=0.45]{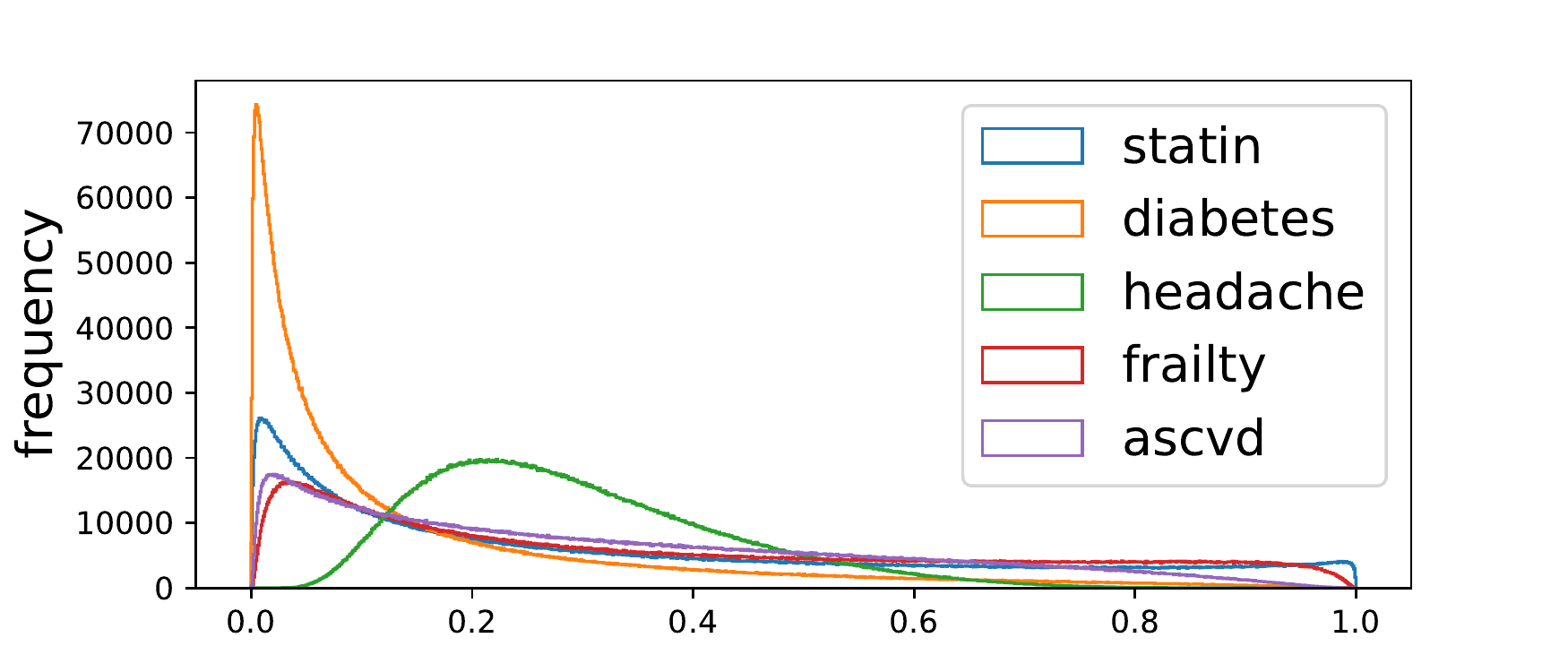} \\
\begin{scriptsize}
\vspace{0.35cm}
\begin{tabular}{l|r|r|r|r|r|r|r|r}
\toprule
&    \textbf{size} & \textbf{age} &  \textbf{diabetes} &  \textbf{headache} &  \textbf{log(pre-LDL)} &  \textbf{log(post-LDL)} &  \textbf{risk score} &     \textbf{ASCVD} \\
\midrule
\textbf{statin}    &    846 &  56.629 &  0.229 &  0.195 &  4.916 &  4.730 &    0.099 &  0.354 \\
\textbf{no statin} &  2154 &  53.092 &  0.020 &  0.042 &  4.860 &  4.878 &    0.076 &  0.301 \\
\bottomrule
\end{tabular}
\end{scriptsize}
\caption{Distributions and statistics of the data generating process} \label{fig:statsmedical}
\end{figure*}
\vspace{-2pt}

\begin{figure*}[h!]
\centering
\includegraphics[scale=0.4]{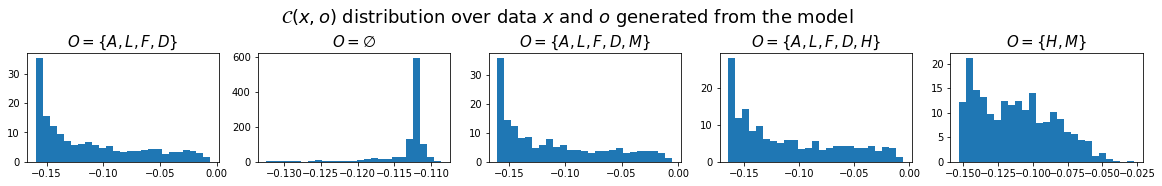} \\
\includegraphics[scale=0.4]{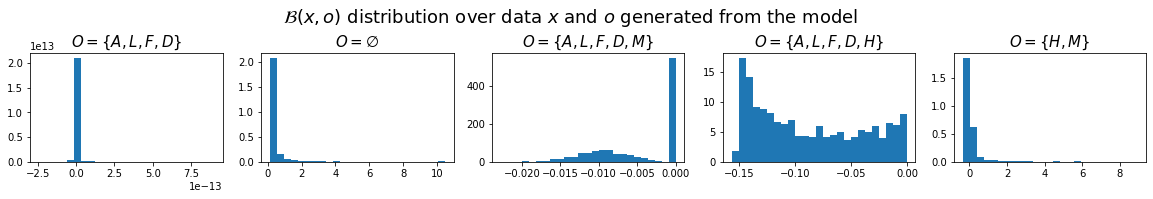} \\
\begin{scriptsize}
\vspace{0.35cm}
\begin{tabular}{l|c|c|c|c|c}
\toprule
& ${\cO=\{A\!,\!L\!,\!F\!,\!D\}}$ & ${\cO=\emptyset}$ &  ${\cO=\{A\!,\!L\!,\!F\!,\!D\!,\!M\}}$ & ${\cO=\{A\!,\!L\!,\!F\!,\!D\!,\!H\}}$ & ${\cO=\{H\!,\!M\}}$ \\
\midrule
\textbf{effect}    &   $-0.112\ \ \ (+0.045)$ & $-0.113\ \ \ (+0.004)$ & $-0.112\ \ \ (+0.045)$ & $-0.112\ \ \ (+0.047)$ & $-0.112\ \ \ (+0.027)$ \\
\textbf{bias}    &   $+0.000\ \ \ (+0.000)$ & $+0.417\ \ \ (+0.895)$ & $-0.005\ \ \ (+0.005)$ & $-0.088\ \ \ (+0.048)$ & $+0.100\ \ \ (+0.608)$ \\
\bottomrule\end{tabular}
\end{scriptsize}
\caption{Distributions and means (standard deviations) of average partial effect on the treated and causal bias}\label{fig:estimatemedical}
\end{figure*}
\vspace{-2pt}

In order to estimate $\cC(x,o)$ and $\cB(x,o)$, here we adopt an importance sampling technique as described in Section \ref{sec:estimation}, choosing the prior $p(U_{\cL})$ as importance density and $N=10^5$. Figure~\ref{fig:estimatemedical} shows distributions of estimated $\cC(x,o)$ and $\cB(x,o)$ with $x$ and $o$ randomly generated from the model, for different sets of observed covariates $\cO$; the table shows their mean estimates, with standard deviation in brackets. When $\cO=\{A,L,F,D\}$ the bias distribution is peaked around zero. This is expected, since all confounders are observed and only them, so the model should exhibit no bias. When $\cO=\emptyset$ we have confounding bias, which we see to be very pronounced. In comparison to the case before, here the distribution of $\cC(x,o)$ is very peaked, showing the large impact of the confounders. This is also the case studied in~\cite{zivich2021machine}, where a true causal effect of $-0.1082$ is reported in the authors' code and is very close to what we estimate for our modified model. When $\cO=\{A,L,F,D, M\}$ we have overcontrol bias, but this does not seem particularly impactful. If $\cO=\{A,L,F,D, H\}$ we have endogenous selection bias, which is smaller than the confounding bias but stronger than the overcontrol one. Finally, if $\cO=\{H,M\}$ all biases interact together. Interestingly, in this case the observation of post-treatment effects such as $H$ and $M$, while themselves contributing to induce bias, partially mitigates the effect of confounding bias. 

\section{Conclusion and future direction}
In this work we developed a new complete criterion for identifiability of causal effects via covariate adjustment. We provided a mathematical relation between association, causal effect, and bias, showing the bias is zero if such criterion holds. We developed a formulation of the bias where all terms can be estimated from data using automatic differentiation. We demonstrated the usefulness of our results over several experiments. A natural future work direction is to exploit the bias characterization as a causal regularizer for the training of highly non-linear machine learning models where we cannot easily control for causal bias.

\bibliography{biblio}

\begin{thebibliography}{46}
\providecommand{\natexlab}[1]{#1}
\providecommand{\url}[1]{\texttt{#1}}
\expandafter\ifx\csname urlstyle\endcsname\relax
  \providecommand{\doi}[1]{doi: #1}\else
  \providecommand{\doi}{doi: \begingroup \urlstyle{rm}\Url}\fi

\bibitem[Abrevaya et~al.(2015)Abrevaya, Hsu, and Lieli]{abrevaya2015estimating}
Abrevaya, J., Hsu, Y.-C., and Lieli, R.~P.
\newblock Estimating conditional average treatment effects.
\newblock \emph{Journal of Business \& Economic Statistics}, 33\penalty0
  (4):\penalty0 485--505, 2015.

\bibitem[Altshuler(1981)]{altshuler1981modeling}
Altshuler, B.
\newblock Modeling of dose-response relationships.
\newblock \emph{Environmental health perspectives}, 42:\penalty0 23--27, 1981.

\bibitem[Angrist \& Pischke(2008)Angrist and Pischke]{angrist2008mostly}
Angrist, J.~D. and Pischke, J.-S.
\newblock \emph{Mostly harmless econometrics: An empiricist's companion}.
\newblock Princeton university press, 2008.

\bibitem[Austin(2019)]{austin2019assessing}
Austin, P.~C.
\newblock Assessing covariate balance when using the generalized propensity
  score with quantitative or continuous exposures.
\newblock \emph{Statistical methods in medical research}, 28\penalty0
  (5):\penalty0 1365--1377, 2019.

\bibitem[Bishop(2006)]{Bishop_2006}
Bishop, C.~M.
\newblock \emph{Pattern Recognition and Machine Learning}.
\newblock Springer, 2006.
\newblock ISBN 978-0-387-31073-2.

\bibitem[Blei et~al.(2017)Blei, Kucukelbir, and McAuliffe]{blei2017variational}
Blei, D.~M., Kucukelbir, A., and McAuliffe, J.~D.
\newblock Variational inference: A review for statisticians.
\newblock \emph{Journal of the American statistical Association}, 112\penalty0
  (518):\penalty0 859--877, 2017.

\bibitem[Brooks et~al.(2011)Brooks, Gelman, Jones, and
  Meng]{brooks2011handbook}
Brooks, S., Gelman, A., Jones, G., and Meng, X.-L.
\newblock \emph{Handbook of Markov Chain Monte Carlo}.
\newblock CRC press, 2011.

\bibitem[Chernozhukov et~al.(2018)Chernozhukov, Chetverikov, Demirer, Duflo,
  Hansen, Newey, and Robins]{chernozhukov2018double}
Chernozhukov, V., Chetverikov, D., Demirer, M., Duflo, E., Hansen, C., Newey,
  W., and Robins, J.
\newblock Double/debiased machine learning for treatment and structural
  parameters, 2018.

\bibitem[Chipman et~al.(2010)Chipman, George, McCulloch,
  et~al.]{chipman2010bart}
Chipman, H.~A., George, E.~I., McCulloch, R.~E., et~al.
\newblock Bart: Bayesian additive regression trees.
\newblock \emph{The Annals of Applied Statistics}, 4\penalty0 (1):\penalty0
  266--298, 2010.

\bibitem[Cinelli \& Hazlett(2020)Cinelli and Hazlett]{cinelli2020making}
Cinelli, C. and Hazlett, C.
\newblock Making sense of sensitivity: Extending omitted variable bias.
\newblock \emph{Journal of the Royal Statistical Society: Series B (Statistical
  Methodology)}, 82\penalty0 (1):\penalty0 39--67, 2020.

\bibitem[Doucet et~al.(2001)Doucet, De~Freitas, and
  Gordon]{doucet2001introduction}
Doucet, A., De~Freitas, N., and Gordon, N.
\newblock An introduction to sequential monte carlo methods.
\newblock In \emph{Sequential Monte Carlo methods in practice}, pp.\  3--14.
  Springer, 2001.

\bibitem[Elwert(2013)]{elwert2013graphical}
Elwert, F.
\newblock Graphical causal models.
\newblock In
  \emph{Handbook
  of causal analysis for social research}, pp.\  245--273. Springer, 2013.

\bibitem[Gain \& Shpitser(2018)Gain and Shpitser]{gain2018structure}
Gain, A. and Shpitser, I.
\newblock Structure learning under missing data.
\newblock In \emph{International Conference on Probabilistic Graphical Models},
  pp.\  121--132. PMLR, 2018.

\bibitem[Galagate(2016)]{galagate2016causal}
Galagate, D.
\newblock \emph{Causal inference with a continuous treatment and outcome:
  Alternative estimators for parametric dose-response functions with
  applications.}
\newblock PhD thesis, 2016.

\bibitem[Gill \& Robins(2001)Gill and Robins]{gill2001causal}
Gill, R.~D. and Robins, J.~M.
\newblock Causal inference for complex longitudinal data: The continuous case.
\newblock \emph{Annals of Statistics}, pp.\  1785--1811, 2001.

\bibitem[Hern{\'a}n \& Robins(2020)Hern{\'a}n and Robins]{hernan2020causal}
Hern{\'a}n, M.~A. and Robins, J.~M.
\newblock Causal inference: What if, 2020.

\bibitem[Hill(2011)]{hill2011bayesian}
Hill, J.~L.
\newblock Bayesian nonparametric modeling for causal inference.
\newblock \emph{Journal of Computational and Graphical Statistics}, 20\penalty0
  (1):\penalty0 217--240, 2011.

\bibitem[Hines et~al.(2021)Hines, Diaz-Ordaz, and
  Vansteelandt]{hines2021parameterising}
Hines, O., Diaz-Ordaz, K., and Vansteelandt, S.
\newblock Parameterising the effect of a continuous exposure using average
  derivative effects, 2021.

\bibitem[Hirano \& Imbens(2004)Hirano and Imbens]{hirano2004propensity}
Hirano, K. and Imbens, G.~W.
\newblock The propensity score with continuous treatments.
\newblock \emph{Applied Bayesian modeling and causal inference from
  incomplete-data perspectives}, 226164:\penalty0 73--84, 2004.

\bibitem[Imai \& Van~Dyk(2004)Imai and Van~Dyk]{imai2004causal}
Imai, K. and Van~Dyk, D.~A.
\newblock Causal inference with general treatment regimes: Generalizing the
  propensity score.
\newblock \emph{Journal of the American Statistical Association}, 99\penalty0
  (467):\penalty0 854--866, 2004.

\bibitem[Imbens(2000)]{imbens2000role}
Imbens, G.~W.
\newblock The role of the propensity score in estimating dose-response
  functions.
\newblock \emph{Biometrika}, 87\penalty0 (3):\penalty0 706--710, 2000.

\bibitem[Imbens \& Rubin(2015)Imbens and Rubin]{imbens2015causal}
Imbens, G.~W. and Rubin, D.~B.
\newblock \emph{Causal inference in statistics, social, and biomedical
  sciences}.
\newblock Cambridge University Press, 2015.

\bibitem[Jordan et~al.(1999)Jordan, Jaakkola, and Saul]{Jordan:1998:vi}
Jordan, M., Jaakkola, and Saul, L.
\newblock An
  Introduction to Variational Methods for Graphical Models.
\newblock \emph{Machine Learning}, \penalty0 (37):\penalty0 183--233, 1999.

\bibitem[Kennedy et~al.(2017)Kennedy, Ma, McHugh, and
  Small]{kennedy2017nonparametric}
Kennedy, E.~H., Ma, Z., McHugh, M.~D., and Small, D.~S.
\newblock Nonparametric methods for doubly robust estimation of continuous
  treatment effects.
\newblock \emph{Journal of the Royal Statistical Society. Series B, Statistical
  Methodology}, 79\penalty0 (4):\penalty0 1229, 2017.

\bibitem[Kingma \& Ba(2014)Kingma and Ba]{kingma2014adam}
Kingma, D.~P. and Ba, J.
\newblock Adam: A method for stochastic optimization.
\newblock \emph{arXiv preprint arXiv:1412.6980}, 2014.

\bibitem[Koller \& Friedman(2009)Koller and Friedman]{koller2009probabilistic}
Koller, D. and Friedman, N.
\newblock \emph{Probabilistic Graphical Models: Principles and Techniques}.
\newblock Adaptive computation and machine learning. MIT Press, 2009.
\newblock ISBN 9780262013192.

\bibitem[Lauritzen(2001)]{lauritzen2001causal}
Lauritzen, S.~L.
\newblock Causal inference from graphical models.
\newblock \emph{Complex stochastic systems}, pp.\  63--107, 2001.

\bibitem[Little \& Rubin(2019)Little and Rubin]{little2019statistical}
Little, R.~J. and Rubin, D.~B.
\newblock \emph{Statistical analysis with missing data}, volume 793.
\newblock John Wiley \& Sons, 2019.

\bibitem[Murphy(2012)]{murphy2012machine}
Murphy, K.~P.
\newblock \emph{Machine learning: A probabilistic perspective}.
\newblock MIT press, 2012.

\bibitem[Neal(2001)]{neal2001annealed}
Neal, R.~M.
\newblock Annealed importance sampling.
\newblock \emph{Statistics and computing}, 11\penalty0 (2):\penalty0 125--139,
  2001.

\bibitem[Papamakarios et~al.(2019)Papamakarios, Nalisnick, Rezende, Mohamed,
  and Lakshminarayanan]{papamakarios2019normalizing}
Papamakarios, G., Nalisnick, E., Rezende, D.~J., Mohamed, S., and
  Lakshminarayanan, B.
\newblock Normalizing flows for probabilistic modeling and inference.
\newblock \emph{arXiv preprint arXiv:1912.02762}, 2019.

\bibitem[Pearl(2009{\natexlab{a}})]{pearl2009}
Pearl, J.
\newblock Causal inference in statistics: {A}n overview.
\newblock \emph{Statistics Surveys}, 3:\penalty0 96--146, 2009{\natexlab{a}}.
\newblock \doi{10.1214/09-SS057}.

\bibitem[Pearl(2009{\natexlab{b}})]{pearl2009causality}
Pearl, J.
\newblock \emph{Causality}.
\newblock Cambridge university press, 2009{\natexlab{b}}.

\bibitem[Pearl(2009{\natexlab{c}})]{reason:Pearl09a}
Pearl, J.
\newblock \emph{Causality: Models, Reasoning and Inference}.
\newblock Cambridge University Press, 2nd edition, 2009{\natexlab{c}}.

\bibitem[Rezende \& Mohamed(2015)Rezende and Mohamed]{rezende2015variational}
Rezende, D. and Mohamed, S.
\newblock Variational inference with normalizing flows.
\newblock In \emph{International Conference on Machine Learning}, pp.\
  1530--1538. PMLR, 2015.

\bibitem[Rosenbaum \& Rubin(1983)Rosenbaum and Rubin]{rosenbaum1983central}
Rosenbaum, P.~R. and Rubin, D.~B.
\newblock The central role of the propensity score in observational studies for
  causal effects.
\newblock \emph{Biometrika}, 70\penalty0 (1):\penalty0 41--55, 1983.

\bibitem[Rosenbaum et~al.(2010)]{rosenbaum2010design}
Rosenbaum, P.~R. et~al.
\newblock \emph{Design of observational studies}, volume~10.
\newblock Springer, 2010.

\bibitem[Shpitser \& Pearl(2006)Shpitser and Pearl]{Shpitser2006}
Shpitser, I. and Pearl, J.
\newblock Identification of joint interventional distributions in recursive
  semi-{M}arkovian causal models.
\newblock In \emph{Proceedings of the 21st National Conference on Artificial
  Intelligence}, volume~2, pp.\  1219--1226. AAAI Press, 2006.

\bibitem[Shpitser et~al.(2010)Shpitser, VanderWeele, and
  Robins]{shpitser2012validity}
Shpitser, I., VanderWeele, T., and Robins, J.~M.
\newblock On the validity of covariate adjustment for estimating causal
  effects.
\newblock In \emph{Proceedings of the 26th Conference on Uncertainty in
  Artificial Intelligence, UAI 2010}, pp.\  527--536. AUAI Press, 2010.

\bibitem[Shun \& McCullagh(1995)Shun and McCullagh]{shun1995laplace}
Shun, Z. and McCullagh, P.
\newblock Laplace approximation of high dimensional integrals.
\newblock \emph{Journal of the Royal Statistical Society: Series B
  (Methodological)}, 57\penalty0 (4):\penalty0 749--760, 1995.

\bibitem[Silva(2016)]{NIPS2016_aff16212}
Silva, R.
\newblock Observational-interventional priors for dose-response learning.
\newblock In Lee, D., Sugiyama, M., Luxburg, U., Guyon, I., and Garnett, R.
  (eds.), \emph{Advances in Neural Information Processing Systems}, volume~29,
  2016.

\bibitem[Spirtes et~al.(2000)Spirtes, Glymour, and Scheines]{Spirtes2000}
Spirtes, P., Glymour, C., and Scheines, R.
\newblock \emph{Causation, Prediction, and Search}.
\newblock MIT press, 2nd edition, 2000.

\bibitem[Tran et~al.(2016)Tran, Ruiz, Athey, and Blei]{tran2016model}
Tran, D., Ruiz, F.~J., Athey, S., and Blei, D.~M.
\newblock Model criticism for Bayesian causal inference.
\newblock \emph{arXiv preprint arXiv:1610.09037}, 2016.

\bibitem[Wang(2015)]{wang2015exposure}
Wang, J.
\newblock \emph{Exposure-Response Modeling: Methods and Practical
  Implementation}, volume~84.
\newblock CRC press, 2015.

\bibitem[Wu et~al.(2018)Wu, Mealli, Kioumourtzoglou, Dominici, and
  Braun]{wu2018matching}
Wu, X., Mealli, F., Kioumourtzoglou, M.-A., Dominici, F., and Braun, D.
\newblock Matching on generalized propensity scores with continuous exposures.
\newblock \emph{arXiv preprint arXiv:1812.06575}, 2018.

\bibitem[Zivich \& Breskin(2021)Zivich and Breskin]{zivich2021machine}
Zivich, P.~N. and Breskin, A.
\newblock Machine learning for causal inference: On the use of cross-fit
  estimators.
\newblock \emph{Epidemiology}, 32\penalty0 (3):\penalty0 393--401, 2021.

\end{thebibliography}
\bibliographystyle{icml2022}

\newpage
\appendix
\onecolumn

\section{Proofs}\label{app:proofs}
\subsection{Proof of Theorem \ref{thm:causal_identification} \label{app:causal_identification}}
Let us first report the adjustment criterion introduced in \cite{shpitser2012validity}, adapted to our framework where, for sake of simplicity, we required $\cX$ and $\cY$ to be singletons. We stress that, unlike the definition in \cite{shpitser2012validity}, by convention in this work we defined parents, ancestors and descendants of a variable not to include the variable itself (see \cite{lauritzen2001causal}). 
\begin{definition}
    Given a DAG $\cG$, a set of nodes $\cO\in \cV\setminus\{X,Y\}$ satisfies the adjustment criterion relative to $(X,Y)$ in $\cG$ if:
    \begin{itemize}
        \item $\cO$ does not include descendants of any variable laying on a causal path from $X$ to $Y$, that is a path from $X$ to $Y$ where all arrows point away from $X$;
        \item $\cO$ blocks all non-causal paths from $X$ to $Y$ in $\cG$.
    \end{itemize}
\end{definition}
The adjustment criterion was shown in \cite{shpitser2012validity} to be complete for identifiability of causal effects via covariate adjustment. We now show that such criterion and \eqref{eq:causal_identification} are equivalent. In the proof, we will use graphical concepts such as blocked/unblocked path, fork and collider; see \cite{pearl2009causality} for their definitions.

Let us proceed by contradiction. Suppose there exists $O_j\in \cO\cap\De(V_i)\cup \{V_i\}$ with $V_i$ laying on a causal path from $X$ to $Y$. Without loss of generality, suppose that a causal path from $X$ to $V_i$ is unblocked, otherwise we could repeat the argument for the first variable that blocks the path instead. Analogously, without loss of generality suppose that a causal path form $V_i$ to $O_j$ is unblocked. Then we have $U_{V_i}\not\indep X|\cO$. Since $U_{V_i}\in \cU^C$, this contradicts \eqref{eq:causal_identification}.

Now suppose instead there exists an unblocked non-causal path from $X$ to $Y$. If the last arrow of such path points away from $Y$, then $U_Y\not\indep X|\cO$. Since $U_Y\in \cU^C$, this contradicts \eqref{eq:causal_identification}. Vice versa, if the last arrow of such path points into $Y$, in order to be unblocked and non-causal the path must contain a fork $V_i\not\in \cO$ such that $V_i\in\An(Y)$. Then $U_{V_i}\not\indep X|\cO$ and $U_{V_i}\in \cU^C$, contradicting \eqref{eq:causal_identification}.

This completes the proof that \eqref{eq:causal_identification} implies the adjustment criterion. Let us now show the other direction. By contradiction, suppose \eqref{eq:causal_identification} does not hold. Then there exists a $U_{V_i}\in \cU\cap\An^+(Y)$ such that $U_{V_i}\not\indep Y|X,\cO\setminus\De(X)$ and $U_{V_i}\not\indep X|\cO$. Note that $V_i\ne X$ because $U_X\indep Y|X,\cO\setminus\De(X)$. Also, we must have $V_i\in\An(Y)$ or $V_i=Y$. 

Suppose that $V_i=Y$. Then there exists some observed variable $O_j\in\De(Y)$ such that a path from $X$ to $O_j$ is unblocked, otherwise $U_Y\indep X|\cO$. Consequently, there exists an unblocked path from $X$ to $Y$. If the latter is non-causal, this contradicts the second item of the adjustment criterion. Otherwise, since also the path from $Y$ to $O_j$ is causal by definition of descendant, we have a contradiction with the first item instead.

Now suppose instead that $V_i\in\An(Y)\setminus \cX$ and $V_i\in \cO$. Since $U_{V_i}\not\indep X|\cO$ there exists an unblocked path from $X$ to $V_i$ with last arrow pointing into $V_i$. If, in addition, $V_i\in \cO\setminus \De(X)$, then there exists an unblocked path from $V_i$ to $Y$ with first arrow into $V_i$, otherwise $U_{V_i}\indep Y|X,\cO\setminus\De(X)$. Then there is an unblocked non-causal path from $X$ to $Y$ with observed collider $V_i$, which contradicts the second item of the adjustment criterion. If, vice versa, $V_i\in \cO\cap\De(X)$, then there is a causal path from $X$ to $V_i$ and one from $V_i$ to $Y$ with $V_i$ observed, which contradicts the first item instead.

On the other hand, suppose $V_i\in\An(Y)\setminus \cX$ and $V_i\not\in \cO$. Since $U_{V_i}\not\indep X|\cO$ there exists an unblocked path from $X$ to $V_i$ with last arrow pointing away from $V_i$. If a causal path from $V_i$ to $Y$ is unblocked, then there is an unblocked non-causal path from $X$ to $Y$ with unobserved fork $V_i$, which contradicts the second item of the adjustment criterion. If a causal path from $V_i$ to $Y$ is blocked by some variable in $\cO\cap \De(X)$, then there is a blocked causal path from $X$ to $Y$, which contradicts the first item. Otherwise, since $U_{V_i}\in \cU^C$ there must be a non-causal path from $V_i$ to $Y$, with first arrow pointing away from $V_i$, that is unblocked when observing variables in $\cO\setminus\De(X)$. If this path is also unblocked by $\cO$, then there is a contradiction with the second item of the adjustment criterion. Otherwise, take the last variable $O_j\in \cO\cap\De(X)$ that is blocking the path. This can only happen if $O_j\in\An(Y)$, because otherwise there would need to be $O_k\in \cO\cap\De(O_j)$ that is either an observed collider or an observed descendant of a collider, which would make the path from $V_i$ to $Y$ blocked when observing variables in $\cO\setminus\De(X)$. Hence we have a contradiction with the first item of the adjustment criterion and conclude the proof.

\subsection{Proof of Proposition \ref{thm:marginal_decomposition}}\label{app:marginal_decomposition}
Conditioned over $X=x$ and $O=o$, we have $Y=f_Y=f_Y^{x,\underline{o}}$, where $=$ stands for equality in distribution. By Definition \ref{def:causal_set}, the causal exogenous set $\cU^C$ uniquely determines $f_Y^{x,\underline{o}}$, as it includes all exogenous random variables that are ancestors of $Y$ and are not independent of $Y$ given $x$ and $\underline{o}$. Then by the Law of the Unconscious Statistician we can write
\[ \bE_{Y|x,o}[Y] = \int y\,p(y|x,o)\,dy = \int f_Y^{x,\underline{o}}\,p(u^c|x,o)\,du^c. \]
Next, as we assumed $f_Y$ differentiable with respect to $X$ and satisfying sufficient regularity assumptions, by integral Leibniz rule we have
\begin{align*}
     \nabla_x \bE_{Y|x,o}[Y] &= \int \nabla_x\left(f_Y^{x,\underline{o}}\,p(u^c|x,o)\right)\,du^c\\
     &=\int \left(\nabla_xf_Y^{x,\underline{o}} + f_Y^{x,\underline{o}}\,\nabla_x\log p(u^c|x,o)\right)p(u^c|x,o)\,du^c\\
     &= \bE_{U^C|x,o}[\nabla_x\,f_Y^{x,\underline{o}}] + \bE_{U^C|x,o}[f_Y^{x,\underline{o}}\,\nabla_x\log p(U^C|x,o)],
\end{align*}
where we used product rule and the fact that $\nabla_x p(u^c|x,o) = p(u^c|x,o)\nabla_x \log p(u^c|x,o)$. Finally, if the causal effect of $X$ on $Y$ can be identified via covariate adjustment then Corollary \ref{cor:grad_p_zero} holds, which trivially implies $\cB(x, o)=0$.

\subsection{Proof of Theorem \ref{thm:bias_as_covariance}}\label{app:bias_as_covariance}
First, we observe that $\nabla_x \log p(U^C|x,o) = \nabla_x \log p(U^C,x,o) - \nabla_x\log p(x,o)$. Then we have
\begin{align*}
    \nabla_x \log p(x,o) &= \frac{1}{p(x,o)}\nabla_x p(x,o)\\
    &= \frac{1}{p(x,o)}\nabla_x \int p(u^c,x,o)\,dc\\
    &= \frac{1}{p(x,o)} \int \nabla_x p(u^c,x,o)\,dc\\
    &= \frac{1}{p(x,o)} \int \nabla_x \log p(u^c,x,o) p(u^c,x,o)\,dc\\
    &= \int \nabla_x \log p(u^c,x,o) p(u^c|x,o)\,dc\\
    &= \bE_{U^C|x,o}[\nabla_x \log p(U^C,x,o)].
\end{align*}
Thus the bias can be rewritten as
\begin{align*}
    \cB(x,o) &= \bE_{U^C|x,o}[f_Y^{x,\underline{o}}\, \nabla_x\log p(U^C|x,o)]\\
    &= \bE_{U^C|x,o}[f_Y^{x,\underline{o}}\, \nabla_x\log p(U^C,x,o)] - \bE_{U^C|x,o}[f_Y^{x,\underline{o}}\, \nabla_x\log p(x,o)]\\
    &= \bE_{U^C|x,o}[f_Y^{x,\underline{o}}\, \nabla_x\log p(U^C,x,o)] - \bE_{U^C|x,o}[f_Y^{x,\underline{o}}]\,\nabla_x\log p(x,o)\\
    &= \bE_{U^C|x,o}[f_Y^{x,\underline{o}}\, \nabla_x\log p(U^C,x,o)] - \bE_{U^C|x,o}[f_Y^{x,\underline{o}}]\,\bE_{C|x,o}[\nabla_x \log p(U^C,x,o)]\\
    &= \Cov_{U^C|x,o}\left(f_Y^{x,\underline{o}}, \nabla_x\log p(U^C,x,o)\right).
\end{align*}

\subsection{Proof of Theorem \ref{thm:bias_rewrite}}\label{app:bias_rewrite}

Let us start by the expression of the bias in \eqref{eq:bias_as_covariance}. We first note that $\bar{f}_Y^{x,\underline{o}}=\int \bar{f}_Y^x p(u_{\underline{O}}|u^c,\underline{o})\,du_{\underline{O}}$, where $p(u_{\underline{O}}|u^c,\underline{o})$ is a Dirac delta distribution. In particular the latter does not depend on $x$ by definition of $\underline{\c O}$. By further marginalization, this also implies that $\bE_{U^C|x,o}[f_Y^{x,\underline{o}}]=\bE_{U|x,o}[f_Y^x]$. For sake of notation, we will denote $\bar{f}_Y^{x,\underline{o}}=f_Y^{x,\underline{o}} - \bE_{U^C|x,o}[f_Y^{x,\underline{o}}]$ and $\bar{f}_Y^x=f_Y^x - \bE_{U|x,o}[f_Y^x]$. Also, denote $u=[u^c,u_{\underline{O}}, \tilde{u}]$, where $u^c$ is a realization of the causal exogenous set and $\tilde{u}$ contains all variables in $u$ that are neither in $u^c$ nor in $u_{\underline{O}}$. We have
\begin{align*}
    \cB(x,o) &= \Cov_{U^C|x,o}(f_Y^{x,\underline{o}}, \nabla_x \log p(u^c,x,o))\\
    &= \bE_{U^C|x,o}[\bar{f}_Y^{x,\underline{o}}\,\nabla_x\log p(u^c,x,o)]\\
    &= \frac{1}{ p(x,o)}\int \bar{f}_Y^{x,\underline{o}}\,\nabla_x p(u^c,x,o)\,du6c\\
    &= \frac{1}{ p(x,o)}\int \left(\int\bar{f}_Y^x\,p(u_{\underline{O}}|u^c,\underline{o})\,du_{\underline{O}}\right)\nabla_x \left(\int  p(u^c,\tilde{u},x,o)\,d\tilde{u}\right)\,du^c\\ 
    &= \frac{1}{ p(x,o)}\int \bar{f}_Y^x\nabla_x (p(u_{\underline{O}}|u^c,\underline{o})\,p(u^c,\tilde{u},x,o))\,du\\
    &= \frac{1}{ p(x,o)}\int \bar{f}_Y^x\nabla_x p(u,x,o)\,du
\end{align*}
where (i) we marginalized $p(u^c,x,o)$ over variables $\tilde{u}$; (ii) we moved the integrals outside, which is possible because $\bar{f}_Y^x$ does not depend on $\tilde{u}$ by definition of $\cU^C$; (iii) we moved $p(u_{\underline{O}}|u^c,\underline{o})$ inside the gradient with respect to $x$, which is allowed since it is independent of it; (iv) we rewrote the density factorization as $p(u,x,o)$, that is correct because $p(u_{\underline{O}}|u^c,\underline{o})=p(u_{\underline{O}}|u^c,x,o, \tilde{u})$ by definition of $\underline{\cO}$. 

Again for sake of notation, let us denote $z=[x,o]$. We indicate by $z_i$ a generic variable $x$ or $o_i$, whereas $z_{-i}$ corresponds to all variables in $z$ except $z_i$. Without loss of generality, we further order the indices $i$ from parents to children, so that a variable $z_i$ may depend on $z_{i-1}$ but not on $z_{i+1}$. We write $z_{<i}=[z_1,\dots,z_{i-1}]$. Then we can factorize
\[ p(u,z) = p(u)\prod_ip(z_i|z_{<i}, u) = p(u)\prod_i\delta(g_{z_i}), \]
where $\delta$ denotes a Dirac delta distribution and we define $g_{z_i}=f_{Z_i}^{x,o}-z_i$. Replacing in the derivation above, we have
\begin{align*}
\frac{1}{ p(x,o)}\int \bar{f}_Y^x\,\nabla_x p(u,x,o)\,du &= \frac{1}{ p(x,o)}\sum_i\int \bar{f}_Y^x\, p(u)p(z_{-i}|u)\nabla_x \delta(g_{z_i})\,du\\
&= \frac{1}{ p(x,o)}\sum_i\int \bar{f}_Y^x\, p(u)p(z_{-i}|u)\nabla_{g_{z_i}} \delta(g_{z_i})\nabla_x g_{z_i}\,du,\\
\end{align*}
where in the last step we used chain rule. Let us now focus on each term of the sum. We decompose $u=[u_{Z_i}, u_{-{Z_i}}]$, where $u_{-{Z_i}}$ contains all variables in $u$ but $u_{Z_i}$. Observe that $p(z_{-i}|u) = p(z_{-i}|u_{-Z_i})$ because $z_{-i}$ is independent of $u_{Z_i}$. We then proceed with a change of variable, i.e.~$p(u_{Z_i})\,du_{Z_i}=p(g_{z_i})\,dg_{z_i}$. Replacing above, we get
\begin{align*} 
&\int \bar{f}_Y^x\, p(u)p(z_{-i}|u)\nabla_{g_{z_i}} \delta(g_{z_i})\nabla_x g_{z_i}\,du\\ 
&=\int \bar{f}_Y^x\, p(u_{Z_i})p(u_{-Z_i})p(z_{-i}|u_{-Z_i})\nabla_{g_{z_i}} \delta(g_{z_i})\nabla_x g_{z_i}\,du_{Z_i}du_{-Z_i}\\ 
&=\int \bar{f}_Y^x\, p(g_{z_i})p(u_{-Z_i})p(z_{-i}|u_{-Z_i})\nabla_{g_{z_i}} \delta(g_{z_i})\nabla_x g_{z_i}\,dg_{z_i}du_{-Z_i}\\ 
&=-\int \nabla_{g_{z_i}}\left(\bar{f}_Y^x\, p(g_{z_i})\nabla_x g_{z_i}p(u_{-Z_i})p(z_{-i}|u_{-Z_i})\right) \delta(g_{z_i})\,dg_{z_i}du_{-Z_i},\\ 
\end{align*}
where in the last step we used the definition of derivative of a Dirac delta, i.e.~$\int h(a)\nabla_a\delta(a)\,da=-\int \nabla_a h(a)\delta(a)\,da$, where $h$ is a smooth function. This is possible because we assumed $f_Y$ to be differentiable in $X$ and $V_i \in \cO$. Note that $\nabla_{g_{z_i}}\nabla_x g_{z_i}=\nabla_x\nabla_{g_{z_i}}g_{z_i}=0$; furthermore, $\nabla_{g_{z_i}}p(u_{-Z_i})=0$ and $\nabla_{g_{z_i}}p(z_{-i}|u_{-Z_i})=0$ because they do not depend on $u_{Z_i}$. Also, notice that
\[ \nabla_{g_{z_i}}p(g_{z_i}) = p(g_{z_i})\nabla_{g_{z_i}}\log p(g_{z_i}) = p(g_{z_i})\left(\nabla_{g_{z_i}}\log p(u_{Z_i}) - \nabla_{g_{z_i}}\log|\det(\nabla_{u_{Z_i}}g_{z_i})|\right). \]
We can always express $\det(\nabla_{u_{Z_i}} g_{z_i}) = \sum_\rho\sign(\rho)\prod_k \partial_{u_{Z_i}}[g_{z_i}]_{\rho(k)}$, where $\rho$ is a permutation over the indices, $\sign(\rho)\in\{+1,-1\}$ is its sign and $[g_{z_i}]_{\rho(k)}$ denotes the ${\rho(k)}$-component of $g_{z_i}$. Because $\partial_{u_{Z_i}}\nabla_{g_{z_i}}[g_{z_i}]_{ \rho(i)}=0$, we conclude that $\nabla_{g_{z_i}}\log|\det(\nabla_{u_{Z_i}} g_{z_i})|) = 0$ and $\nabla_{g_{z_i}}p(g_{z_i})=p(g_{z_i})\nabla_{g_{z_i}}\log p(u_{Z_i})$. By replacing in the derivation above, we have
\begin{align*}
    &-\int \nabla_{g_{z_i}}\left(\bar{f}_Y^x\, p(g_{z_i})\nabla_x g_{z_i}\,p(u_{-Z_i})p(z_{-i}|u_{-Z_i})\right) \delta(g_{z_i})\,dg_{z_i}du_{-Z_i}\\
    &=-\int \left(\nabla_{g_{z_i}}\bar{f}_Y^x + \bar{f}_Y^x\,\nabla_{g_{z_i}}\log p(u_{Z_i})\right)\nabla_x g_{z_i}\,p(g_{z_i})p(u_{-Z_i})p(z_{-i}|u_{-Z_i}) \delta(g_{z_i})\,dg_{z_i}du_{-Z_i}\\
    &=-\int \left(\nabla_{g_{z_i}}\bar{f}_Y^x + \bar{f}_Y^x\,\nabla_{g_{z_i}}\log p(u_{Z_i})\right)\nabla_x g_{z_i}\,p(u_{Z_i})p(u_{-Z_i})p(z_{-i}|u_{-Z_i}) \delta(g_{z_i})\,du\\
    &=-\int \left(\nabla_{g_{z_i}}\bar{f}_Y^x + \bar{f}_Y^x\,\nabla_{g_{z_i}}\log p(u_{Z_i})\right)\nabla_x g_{z_i}\,p(u, z)\,du,\\
\end{align*}
where in the last two steps we undid the change of variables and reassembled the joint density. We now sum together all terms over $i$, scale by $p(z)$ and remember that we defined $z=[x,o]$, $\bar{f}_Y^x=f_Y^x-\bE_{U|x,o}[f_Y^x]$ and $g_{v_i}=f_{V_i}^{x,o}-v_i$ for $V_i\in \cX\cup \cO$. We have
\begin{align*}
    \cB(x,o)&=-\frac{1}{ p(z)}\sum_i\int \left(\nabla_{g_{z_i}}\bar{f}_Y^x + \bar{f}_Y^x\,\nabla_{g_{z_i}}\log p(u_{Z_i})\right)\nabla_x g_{z_i}\,p(u, z)\,du\\
    &=-\sum_i\int \left(\nabla_{g_{z_i}}\bar{f}_Y^x + \bar{f}_Y^x\,\nabla_{g_{z_i}}\log p(u_{Z_i})\right)\nabla_x g_{z_i}\,p(u| z)\,du\\  
    &=-\sum_{V_i\in X\cup O}\int \left(\nabla_{g_{v_i}}\bar{f}_Y^x + \bar{f}_Y^x\,\nabla_{g_{v_i}}\log p(u_{V_i})\right)\nabla_x g_{v_i}\,p(u| x,o)\,du\\    
    &=-\sum_{V_i\in X\cup O}\bE_{U|x,o}\left[\left(\nabla_{f_{V_i}^{x,o}}f_Y^x + \left(f_Y^x-\bE_{U|x,o}[f_Y^x]\right)\nabla_{f_{V_i}^{x,o}}\log p(u_{V_i})\right)\nabla_x (f_{V_i}^{x,o}-v_i)\right],
\end{align*}
where we named $V_i$ to be a generic variable in $\cX\cup \cO$ and we expressed the integral as an expectation under the probability density $p(u|x,o)$. We can finally observe that for a generic function $h$ differentiable in $f_{V_i}^{x,o}$, by chain rule we have $\nabla_{u_{V_i}}h = \nabla_{f_{V_i}^{x,o}}h \nabla_{u_{V_i}}f_{V_i}^{x,o}$. Because $f_{V_i}^{x,o}$ is invertible with respect to $U_{V_i}$ for $V_i\in \cX\cup \cO$, the Jacobian $\nabla_{u_{V_i}}f_{V_i}^{x,o}$ is also invertible and we can write $\nabla_{f_{V_i}^{x,o}}h = \nabla_{u_{V_i}}h\left(\nabla_{u_{V_i}}f_{V_i}^{x,o}\right)^{-1}$. Replacing in the derivation above, we conclude that
    \begingroup\makeatletter\def\f@size{9.5}\check@mathfonts
    \def\maketag@@@#1{\hbod{\m@th\normalsize\normalfont#1}}
    \[
    \cB(x,o) = -\!\!\!\sum_{V_i\in \cX\cup \cO}\!\!\bE_{U|x,o}\Big[\Big(\nabla_{u_{V_i}} f_Y^x + \left(f_Y^x-\bE_{U|x,o}[f_Y^x]\right)\nabla_{u_{V_i}}\!\log p(U_{V_i}) \Big)(\nabla_{u_{V_i}}f_{V_i}^{x,o})^{-1}\, \nabla_x (f_{V_i}^{x,o}  - v_i)\Big].
    \]
    \endgroup

\section{Confounding, overcontrol and endogenous selection bias}\label{app:bias_types}
\subsection{Confounding bias}\label{app:confounding}
Let us consider the model in Figure \ref{fig:confounding_model}. Following \eqref{eq:marginal_causal_effect}, the average partial effect on the treated can be expressed as
\[ \cC(x,o)=\bE_{U|x,o}[\nabla_x f_Y^x] = \bE_{U|x,o}[\beta]=\beta. \]
For the marginal causal bias, we have $\nabla_{u_X}f_Y^{x,o}=0$, $\nabla_{u_X}\log p(U_X) = -U_X$, $\nabla_{u_X}f_X^{x,o}=1$ and $\nabla_x(f_X^{x,o} - x) = -1$. Then \eqref{eq:bias_rewrite} gives
\[
    \cB(x,o) = -\bE_{U|x,o}[\left(f_Y^x-\bE_{U|x,o}[f_Y^x]\right)U_X].
\]
Notice that $p(u_{V_1},u_X,u_Y|x)=p(u_{V_1},u_X|x)p(u_Y)$, which implies that $U_Y|X=x$ and $U_X|X=x$ are independent. This yields
\begin{align*} 
&-\bE_{U|x,o}[\left(f_Y^x-\bE_{U|x,o}[f_Y^x]\right)U_X]\\
&=-\bE_{U|x,o}[(\beta x + U_Y - \bE_{U|x,o}[\beta x + U_Y])U_X] - \bE_{U|x,o}[(\gamma U_{V_1} - \bE_{U|x,o}[\gamma U_{V_1}])U_X] \\
&= -\bE_{U|x,o}[\beta x + U_Y - \bE_{U|x,o}[\beta x + U_Y]]\,\bE_{U|x,o}[U_X] - \gamma\,\bE_{U|x,o}[(U_{V_1} - \bE_{U|x,o}[U_{V_1}])U_X]\\
&=-\gamma\,\bE_{U|x,o}[(U_{V_1} - \bE_{U|x,o}[U_{V_1}])U_X]\\
&=-\gamma\,\Cov_{U|x,o}(U_{V_1}, U_X),
\end{align*}
since $\bE_{U|x,o}[\beta x + U_Y - \bE_{U|x,o}[\beta x + U_Y]]=0$. With some simple calculation, one can show that $p(u_X|x)=N(\tfrac{x}{1+\alpha^2}, \frac{\alpha^2}{1+\alpha^2})(u_X)$ and $p(u_{V_1}|x) = N(\tfrac{\alpha x}{1 + \alpha^2}, \tfrac{1}{1+\alpha^2})(u_{V_1})$, where $N(\mu,\sigma^2)(a)$ generally denotes a Gaussian probability density with mean $\mu$, variance $\sigma^2$, evaluated at $a$. We then have that $\bE_{U|x,o}[U_X]=\tfrac{x}{1+\alpha^2}$, $\bE_{U|x,o}[U_{V_1}]=\tfrac{\alpha x}{1 + \alpha^2}$ and $\Var_{U|x,o}(U_{V_1})=\tfrac{1}{1+\alpha^2}$. Furthermore, we have
\begin{align*}
\bE_{U|x,o}[U_{V_1} U_X] &= \int u_{V_1} u_X p(u_{V_1},u_X|x)\,du_X\,du_{V_1}\\
&=\int u_{V_1} \int u_X p(u_X|u_{V_1},x)\,du_X\, p(u_{V_1}|x)\,du_{V_1}\\
&=\int u_{V_1} \int u_X \delta(u_X+\alpha u_{V_1} - x)\,du_X\, p(u_{V_1}|x)\,du_{V_1}\\
&=\int u_{V_1} (x - \alpha u_{V_1})\,p(u_{V_1}|x)du_{V_1}\\
&= x\bE_{U|x,o}[U_{V_1}] - \alpha \bE_{U|x,o}[U_{V_1}^2]\\
&=x\bE_{U|x,o}[U_{V_1}] - \alpha (\Var_{U|x,o}(U_{V_1}) + \bE_{U|x,o}[U_{V_1}]^2).
\end{align*}
Then we have all the elements to compute
\[ \cB(x,o)=-\gamma\left(\bE_{U|x,o}[U_{V_1} U_X] - \bE_{U|x,o}[U_{V_1}]\bE_{U|x,o}[U_X]\right) = \frac{\gamma\alpha}{1+\alpha^2}. \]
Thus
\[ \cC(x,o) + \cB(x,o) = \beta + \frac{\gamma\alpha}{1+\alpha^2}.\]
Let us check that the latter matches the association $\cA(x,o)$ when computed directly. We have
\begin{align*}
    \bE_{Y|x,o}[Y] &= \beta x + \gamma\bE_{V_1|x,o}[V_1] + \bE_{U|x,o}[U_Y]\\
    &= \beta x + \gamma\bE_{U|x,o}[U_{V_1}]\\
    &= \beta x + \frac{\gamma\alpha}{1+\alpha^2}x.
\end{align*}
Then $\cA(x,o)=\nabla_x \bE_{Y|x,o}[Y] = \beta + \tfrac{\gamma\alpha}{1+\alpha^2}$, which verifies the statement.

\subsection{Overcontrol bias}\label{app:overcontrol}
Let us consider the model in Figure \ref{fig:overcontrol_model}. Following \eqref{eq:marginal_causal_effect}, we have
\[ \cC(x,o)= \bE_{U|x,o}[\beta + \gamma\alpha] = \beta + \gamma\alpha. \]
For the bias, we have $\nabla_{u_X}f_Y^{x,o}=0$, $\nabla_{u_X}\log p(U_X)=-U_X$, $\nabla_{u_X}f_X^{x,o}=1$ and $\nabla_x (f_X^{x,o}-x)=-1$. Then the contribution given by $X$ to the bias is null, since
\[
    \bE_{U|x,o}[\left(f_Y^x-\bE_{U|x,o}[f_Y^x]\right)U_X]=\bE_{U|x,o}[f_Y^x-\bE_{U|x,o}[f_Y^x]]\bE_{U|x,o}[U_X]= 0,
\]
where we used that $f_Y^x$ is independent of $U_X$ given $X=x$.

Regarding the contribution to the bias given by $V_1$, we have $\nabla_{u_{V_1}}f_Y^x=\gamma$, $\nabla_x f_{V_1}^{x,o}=\alpha$ and $\nabla_{u_{V_1}}f_{V_1}^{x,o}=1$. Moreover, $U_Y|V_1$ is independent of $X$. Then
\begin{align*}
    \cB(x,o)
    &= -\gamma\alpha + \alpha\bE_{U|x,o}[\left(f_Y^x-\bE_{U|x,o}[f_Y^x]\right)U_{V_1}]\\
    &= -\gamma\alpha + \alpha\bE_{U|x,o}[\beta x + \gamma v_1 + U_Y - \bE_{U|x,o}[\beta x + \gamma v_1 + U_Y]]\bE_{U|x,o}[U_{V_1}]\\
    &= -\gamma\alpha
\end{align*}
since $\bE_{U|x,o}[\beta x + \gamma v_1 + U_Y - \bE_{U|x,o}[\beta x + \gamma v_1 + U_Y]]=0$. Then
\[ \cC(x,o)+\cB(x,o) = \beta + \gamma\alpha -\gamma\alpha = \beta. \]
Let us check the result by computing the marginal association directly. We have $p(y|x,v_1)=N(\beta x + \gamma v_1, 1)(y)$, whence $\bE_{Y|x,o}[Y]=\beta x + \gamma v_1$ and $\cA(x,o)=\nabla_x \bE_{Y|x,o}[Y] = \beta$.

\subsection{Endogenous selection bias}\label{app:endogenous_selection}
Let us consider the model in Figure \ref{fig:selection_model}. Following \eqref{eq:marginal_causal_effect}, we have
\[ \cC(x,o)=\bE_{U|x,o}[\nabla_x f_Y^x] = \bE_{U|x,o}[\alpha] = \alpha. \]
For the bias, we can argue that the contribution of $X$ is null for exactly the same reason as in the overcontrol bias case (see Appendix \ref{app:overcontrol}). Regarding the contribution of $V_1$, we have $\nabla_{u_{V_1}}f_Y^x=0$, $\nabla_{u_{V_1}}\log p(U_{V_1})=-U_{V_1}$, $\nabla_{u_{V_1}}f_{V_1}^{x,o}=1$ and $\nabla_x (f_{V_1}^{x,o}-{v_1})=\beta + \gamma\alpha$. Then
\begin{align*}
    \cB(x,o) &= (\beta + \gamma\alpha)\bE_{U|x,o}[\left(f_Y^x-\bE_{U|x,o}[f_Y^x]\right)U_{V_1}]\\
    &= (\beta + \gamma\alpha)\bE_{U|x,o}[\left(U_Y-\bE_{U|x,o}[U_Y]\right)U_{V_1}]\\
    &= (\beta + \gamma\alpha)\Cov_{U|x,o}(U_Y, U_{V_1}).
\end{align*}
One can compute $p(u_Y|x,v_1)=N(\tfrac{\gamma}{1+\gamma^2}(v_1-(\beta+\gamma\alpha)x),\tfrac{1}{1+\gamma^2})(u_Y)$ and $p(u_{V_1}|x,v_1)=N(\tfrac{1}{1+\gamma^2}(v_1-(\beta + \gamma\alpha)x), \tfrac{\gamma^2}{1+\gamma^2})(u_{V_1})$, which respectively imply $\bE_{U|x,o}[U_Y]=\tfrac{\gamma}{1+\gamma^2}(v_1-\beta-\gamma\alpha)x$,  $\bE_{U|x,o}[U_{V_1}]=\tfrac{x}{1+\gamma^2}(v_1-\beta - \gamma\alpha)$ and $\Var_{U|x,o}(U_{V_1})=\tfrac{\gamma^2}{1+\gamma^2}$. If $\gamma=0$, then $\bE_{U|x,o}[U_{V_1} U_Y]=0$. Otherwise
\begin{align*}
\bE_{U|x,o}[U_{V_1} U_Y] &= \int u_{V_1} u_Y p(u_{V_1},u_Y|x,v_1)\,du_Y\,du_{V_1}\\
&=\int u_{V_1} \int u_Y p(u_Y|u_{V_1},x,v_1)\,du_Y\, p(u_{V_1}|x,v_1)\,du_{V_1}\\
&=\int u_{V_1} \int u_Y \delta((\beta + \gamma\alpha)x +\gamma u_Y + u_{V_1} - v_1)\,du_Y\, p(u_{V_1}|x,v_1)\,du_{V_1}\\
&=\int u_{V_1} \left(\frac{1}{\gamma}(v_1 - (\beta+\gamma\alpha)x) - \frac{1}{\gamma} u_{V_1}\right)\,du_{V_1}\\
&= \frac{1}{\gamma}(v_1 - (\beta+\gamma\alpha)x)\bE_{U|x,o}[U_{V_1}] - \frac{1}{\gamma} \bE_{U|x,o}[U_{V_1}^2]\\
&=\frac{1}{\gamma}(v_1 - (\beta+\gamma\alpha)x)\bE_{U|x,o}[U_{V_1}] - \frac{1}{\gamma} (\Var_{U|x,o}(U_{V_1}) + \bE_{U|x,o}[U_{V_1}]^2).
\end{align*}
Then we have all the elements to compute
\[ \cB(x,o)=(\beta+\gamma\alpha)\left(\bE_{U|x,o}[U_{V_1} U_Y] - \bE_{U|x,o}[U_{V_1}]\bE_{U|x,o}[U_Y]\right)=-\frac{\gamma(\beta+\gamma\alpha)}{1+\gamma^2}. \]
Thus
\[ \cC(x,o) + \cB(x,o) = \frac{\alpha - \gamma\beta}{1 + \gamma^2}. \]
Once more, we check that this is the same as computing the association directly. Since $p(y|x,v_1)=N(\tfrac{1}{1+\gamma^2}\left(\gamma v_1 + (\alpha - \beta\gamma)x\right), \tfrac{1}{1+\gamma^2})(y)$, we have $\bE_{Y|x,o}[Y]=\tfrac{1}{1+\gamma^2}\left(\gamma v_1 + (\alpha - \beta\gamma)x\right)$, whence $\cA(x,o)=\nabla_x \bE_{Y|x,o}[Y]=\tfrac{\alpha - \beta\gamma}{1+\gamma^2}$.

\section{Simulated study of statins and atherosclerotic cardiovascular disease}\label{app:ASCVD}
The model is given by
\begin{align*}
    A &= U_A\\
    L &= \theta^L_0 A + \theta^L_1 + e^{\theta^L_2}U_L\\
    F &= \text{sigmoid}\left(\theta^F_0 + \theta^F_1(A + \theta^F_2) + \theta^F_3 A^2 + e^{\theta^F_4}U_F\right)\\
    D &= \text{sigmoid}\left(\theta^D_0 + \theta^D_1 L + \theta^D_2 A + \theta^D_3 A^2 + e^{\theta^D_4} u_D\right)\\
    R &= \text{sigmoid}\left(\theta^R_0 + \theta^R_1 D + \theta^R_2 \log(A) + \theta^R_3 \log^2(A) + \theta^R_4 L + \theta^R_5 \log(A) L + \theta^R_6 F\right)\\
    X &= \text{sigmoid}\Big(\theta^X_1 \mathbbm{1}_{0.05 \le R < 0.075} + \theta^X_2 \mathbbm{1}_{0.075 \le R < 0.2} + \theta^X_3 \mathbbm{1}_{R \ge 0.2} + \theta^X_4 + \theta^X_5 \mathbbm{1}_{D \ge 0.5}\\ &\hspace{1.7cm}+ \theta^X_6 L + \theta^X_7 \mathbbm{1}_{L > \log(160)} + \theta^X_8(A + \theta^X_0) + \theta^X_9(A + \theta^X_0)^2 + e^{\theta^X_{10}}U_X\Big)\\
    M &= L + \theta^M_0 + \theta^M_1 X (\theta^M_2 - L)\mathbbm{1}_{L < \log(130)} + e^{\theta^M_3}U_M\\
    Y &= \text{sigmoid}\Big(\theta^Y_3 + \theta^Y_4 X + \theta^Y_5 M + \theta^Y_6 \sqrt{A + \theta^Y_0} + \theta^Y_7 D + \theta^Y_8 e^{1 + R}\\ &\hspace{1.7cm}+ \theta^Y_9 \left(L + \theta^Y_1 \text{sigmoid}\left(10(L+\theta^Y_2)\right)L^2\right) + e^{\theta^Y_{10}}U_Y\Big)\\
    H &= \text{sigmoid}\left(\theta^H_0 + \theta^H_1 X + \theta^H_2 Y + e^{\theta^H_3}U_H\right),
\end{align*}
with 
\begin{align*}
    \theta^L&=[0.005, \log(100), \log(0.18)]\\
    \theta^F&=[-5.5, 0.05, -20, 0.001, \log(1.1)]\\
    \theta^D&=[-4.23, 0.03, -0.02, 0.0009, \log(1.6)]\\
    \theta^R&=[4.3, 3.5, -2.07, 0.05, 4.09, -1.04, 0.01]\\
    \theta^X&=[-30, 0.273, 1.592, 2.461, -3.471, 1.39, 0.112, 0.973, -0.046, 0.003, \log(1.7)]\\
    \theta^M&=[0.1, -3.5, 5, 0]\\
    \theta^Y&=[-39, 1.4, -\log(110), -6.25, -0.75, -0.1, 0.45, 1.75, 0.29, 0.1, \log(0.9)]\\
    \theta^H&=[-1.7, 0.8, 1.5, \log(0.5)].
\end{align*}
We take $p(U_A)$ to be a trapezoidal distribution with bottom base from 40 to 75 and top base from 40 to 60. All other distributions of exogenous random variables are taken to be standard Gaussians, that is $U_L, U_F, U_D, U_X, U_M, U_Y, U_H\sim \cN(0, 1)$.

\end{document}